\documentclass[a4paper]{article}
\usepackage[utf8]{inputenc}

\usepackage{geometry}
\usepackage{moreverb}

\usepackage{amsthm}  %
\usepackage{graphicx}
\usepackage{xcolor}
\usepackage{cite}

\usepackage{algorithm}
\usepackage{algpseudocode}
\algrenewcommand{\alglinenumber}[1]{\small #1:}

\newcommand{\algfontsize}{\footnotesize}

\newcommand{\msg}[1]{$\langle #1 \rangle$}
\newcommand{\mtag}[1]{\textit{#1}}
\newcommand{\proc}[1]{\textsf{#1}}
\newcommand{\KwTrue}{true}
\newcommand{\KwFalse}{false}
\newcommand{\KwNull}{null}

\algrenewcommand{\algorithmiccomment}[1]{// #1}

\newtheorem{Theorem}{Theorem}
\newtheorem{Lemma}{Lemma}
\newtheorem{Define}{Definition}

\title{A cooperative partial snapshot algorithm for checkpoint-rollback recovery of large-scale and dynamic distributed systems and experimental evaluations%
\protect\thanks{A preliminary version of this paper appeared in the proceedings of the Sixth International Symposium on Computing and Networking Workshops (CANDARW)\cite{Kim2018}.\\%
This is the peer reviewed version of the following article \cite{Nakamura2020}, which has been published in final form at https://doi.org/10.1002/cpe.5647. This article may be used for non-commercial purposes in accordance with Wiley Terms and Conditions for Use of Self-Archived Versions.}
}

\usepackage{authblk}

\author[1]{Junya Nakamura\thanks{Corresponding author: junya[at]imc.tut.ac.jp}}
\author[2]{Yonghwan Kim}
\author[2]{Yoshiaki Katayama}
\author[3]{Toshimitsu Masuzawa}

\affil[1]{Toyohashi University of Technology, Japan}
\affil[2]{Nagoya Institute of Technology, Japan}
\affil[3]{Osaka University, Japan}

\begin{document}

\maketitle

\begin{abstract}
A distributed system consisting of a huge number of computational entities is prone to faults, because faults in a few nodes cause the entire system to fail.
Consequently, fault tolerance of distributed systems is a critical issue.
Checkpoint-rollback recovery is a universal and representative technique for fault tolerance;
it periodically records the entire system state (configuration) to non-volatile storage, and the system restores itself using the recorded configuration when the system fails.
To record a configuration of a distributed system, a specific algorithm known as a snapshot algorithm is required.
However, many snapshot algorithms require coordination among all nodes in the system; thus, frequent executions of snapshot algorithms require unacceptable communication cost, especially if the systems are large.
As a sophisticated snapshot algorithm, a partial snapshot algorithm has been introduced that takes a partial snapshot (instead of a global snapshot).
However, if two or more partial snapshot algorithms are concurrently executed, and their snapshot domains overlap, they should coordinate, so that the partial snapshots (taken by the algorithms) are consistent.
In this paper, we propose a new efficient partial snapshot algorithm with the aim of reducing communication for the coordination.
In a simulation, we show that the proposed algorithm drastically outperforms the existing partial snapshot algorithm, in terms of message and time complexity.
\end{abstract}

\section{Introduction}
A distributed system consists of computational entities (i.e., computers), usually called nodes, which are connected to each other by (communication) links.
Each node can communicate with the other nodes by exchanging messages through these links.
In large-scale distributed systems, node faults are inevitable, and the faults of only a few nodes (probably a single node) may cause the entire system to fail.
Therefore, the fault tolerance of distributed systems is a critical issue to ensure system dependability.

\emph{Checkpoint-rollback recovery} \cite{Koo1987} is a universal and representative method for realizing the fault tolerance of distributed systems.
Each node periodically (or when necessary) records its local state in non-volatile storage, from which the node recovers its past non-faulty state when faults occur.
This recorded state is called a \emph{checkpoint} and restoring the node state using its checkpoint is called a \emph{rollback}.
However, in distributed systems, to guarantee consistency after a rollback (i.e., a global state constructed from the checkpoints), nodes must cooperate with each other to record their checkpoints.
A configuration is \emph{inconsistent}\cite{Netzer1995,Fischer1982} if it contains an \emph{orphan message}, which is received but is not sent in the configuration.
To resolve the inconsistency, the receiver of the orphan message must restore an older checkpoint.
This may cause a domino effect\cite{Briatico1984} of rollbacks, which is an unbounded chain of local restorings to attain a consistent global state.
A consistent global state can be formed by every node's mutually concurrent local state (which means that there are no causal relationships between any two local states in the global state) and all in-transit messages.
A \emph{snapshot algorithm} is for recording a consistent global configuration called a \emph{snapshot} which ensures that all nodes record their checkpoints cooperatively.
Checkpoint-rollback recovery inherently contains a snapshot algorithm to record the checkpoints of the nodes, forming a consistent global state, and its efficiency strongly depends on that of the snapshot algorithm.

Many sophisticated snapshot algorithms have been proposed \cite{Spezialetti1986,Prakash1994,Elnozahy2002,Moriya2005,Kim2011}.
As the scale (the number of nodes) of a distributed system increases, the efficiency of the snapshot algorithm becomes more important.
Especially in a large-scale distributed system, frequent captures of global snapshots incur an unacceptable communication cost.
To resolve the problem of global snapshot algorithms, \emph{partial snapshot} algorithms have been proposed, which take a snapshot of some portion of a distributed system, rather than the entire system.
Most snapshot algorithms (whether global or partial) cannot deal with dynamic distributed systems where nodes can freely join and leave the system at any time.

In this paper, we propose a new cooperative partial snapshot algorithm which (a) takes a partial snapshot of the communication-related subsystem (called a \emph{snapshot group}), so its message complexity does not depend on the total number of nodes in the system; (b) allows concurrent initiations of the algorithm by two or more nodes, and takes a consistent snapshot using elaborate coordinations among the nodes with a low communication cost; and (c) is applicable to dynamic distributed systems.
Our simulation results show that the proposed algorithm succeeds in drastically decreasing the message complexity of the coordinations compared with previous works.

The rest of this paper is organized as follows: Section \ref{sec:related-works} introduces related work.
Section \ref{sec:preliminaries} presents the system model and details of a previous work on which our algorithm is based.
The proposed algorithm designed to take concurrent partial snapshots and detect the termination is described in Section \ref{sec:proposed-algorithm}.
Section \ref{sec:correctness} discusses the correctness of the algorithm.
The performance of the algorithm is experimentally evaluated in comparison with that of an existing algorithm in Section \ref{sec:evaluation}.
Finally, Section \ref{sec:conclusion} concludes the paper.

\section{Related Work}
\label{sec:related-works}
Chandy and Lamport \cite{Chandy1985} proposed a distributed snapshot algorithm that takes a global snapshot of an entire distributed system.
This global snapshot algorithm ensures its correctness when a distributed system is static: No node joins or leaves, and no (communication) link is added or removed.
Moreover, the algorithm assumes that all links guarantee the First in First out (FIFO) property, and each node knows its neighbor nodes.
Chandy and Lamport's snapshot algorithm uses a special message named \emph{Marker}, and each node can determine the timing to record its own local state using the \mtag{Marker} message.
Some snapshot algorithms for distributed systems with non-FIFO links have also been proposed \cite{Lai1987}.
These global snapshot algorithms are easy to implement and take a snapshot of the distributed system.
However, the algorithms require $\mathcal{O}(m)$ messages (where $m$ is the number of links), because every pair of neighboring nodes has to exchange \mtag{Marker} messages.
Therefore, these algorithms are not practically applicable to large-scale distributed systems which consist of a huge number of nodes.

Some researchers have tried to reduce the number of messages of snapshot algorithms\cite{Kshemkalyani2010,Garg2006,Garg2010}, e.g., $\mathcal{O}(n \log n)$, but the complexity depends on $n$, the number of nodes in the entire system.
This implies that the scalability of snapshot algorithms remains critical.
Not only the scalability problem but also applicability to dynamic distributed systems (where nodes can join and leave the distributed system at any time) are important for global snapshot algorithms.

An alternative approach to scalable snapshot algorithms called communication-induced checkpointing has been studied \cite{Helary1999,Baldoni1998,Baldoni1997,Elnozahy2002}.
In this approach, not all nodes are requested to record their local states (as their checkpoints), but some are, depending on the communication pattern.
For distributed applications mainly based on local coordination among nodes, communication-induced checkpoint algorithms can reduce the communication and time required for recording the nodes' checkpoints.
However, these algorithms cannot guarantee that the \emph{latest} checkpoints of the nodes form a consistent global state.
This forces each node to keep multiple checkpoints in the node's non-volatile storage, and requires an appropriate method to find a set of node checkpoints that forms a consistent global state.
Thus, from a practical viewpoint, these snapshot algorithms cannot solve the scalability problem.

Moriya and Araragi \cite{Moriya2001,Moriya2005} introduced a partial snapshot
\footnote{In \cite{Spezialetti1986}, they called a portion of a global snapshot a partial snapshot; however, the notion of a partial snapshot is different from that in our algorithm, SSS algorithm \cite{Moriya2001,Moriya2005}, and CSS algorithm \cite{Kim2011,Kim2014}.  In this paper, a partial snapshot is not a part of a global snapshot, but a snapshot of a subsystem.}
algorithm, which takes a snapshot of the subsystem consisting only of communication-related nodes, named Sub-SnapShot (SSS) algorithm.
They also proved that the entire system can be restored from faults, using the latest checkpoint of each node.
A communication-related subsystem can be transitively determined by the communication-relation, which is dynamically created by (application layer) communications (exchanging messages) among the nodes.
In practical distributed systems, the number of nodes in a communication-related subsystem is expected to be much smaller than the total number of nodes in the distributed system.
This implies that the number of messages required for SSS algorithm does not depend on the total number of nodes.
Therefore, SSS algorithm can create checkpoints efficiently, so that SSS algorithm makes the checkpoint-rollback recovery applicable to large-scale distributed systems.
However, SSS algorithm cannot guarantee the consistency of the (combined) partial snapshot, if two or more nodes concurrently initiate SSS algorithm instances, and their snapshot groups (communication-related subsystems) overlap.

Spezialetti \cite{Spezialetti1986} presented snapshot algorithms to allow concurrent initiation of two or more snapshot algorithms, and an improved variant was proposed by Prakash \cite{Prakash1994}.
However, their algorithms still target the creation of a global snapshot, and their algorithms are not applicable to dynamic distributed systems.
SSS algorithm is applicable to dynamic distributed systems, where nodes can join and leave the system freely, because the algorithm uses only the communication-relation, which changes dynamically, and requires no a priori knowledge about the topology of the entire system.

Another snapshot algorithm for dynamic distributed systems was introduced by Koo and Toueg \cite{Koo1987}.
However, this communication-induced checkpoint algorithm has to suspend executions of all applications while taking a snapshot, to guarantee the snapshot's consistency.
In contrast, SSS algorithm allows execution of any applications while a snapshot is taken, with some elaborate operations based on the communication-relation.

Kim et al., proposed a new partial snapshot algorithm, named Concurrent Sub-Snapshot (CSS) algorithm \cite{Kim2011,Kim2014}, based on SSS algorithm.
They called the problematic situation caused by the overlap of the subsystems a \emph{collision} and presented an algorithm that can resolve collisions by combining colliding SSS algorithm instances.
In CSS algorithm, to resolve the collision, leader election among the initiating nodes of the collided subsystems is executed, and only one leader node becomes a coordinator.
The coordinator and the other initiators are called the \emph{main-initiator} and \emph{sub-initiators}, respectively.
This leader election is executed repeatedly, to elect a new coordinator when a new collision occurs.
All sub-initiators forward all information collected about the subsystems to the main-initiator, so that all the snapshot algorithm instances are coordinated to behave as a single snapshot algorithm which is initiated by the main-initiator.

CSS algorithm successfully realizes an efficient solution for the collision problem, by consistently combining two or more concurrent SSS algorithm executions.
However, if a large number of nodes concurrently initiate CSS algorithm instances, and the nodes collide with each other many times, leader elections are executed concurrently and repeatedly, and an enormous number of messages are forwarded to the main-initiator.
This overhead for combining snapshot groups and forwarding snapshot information for coordination is the most critical drawback of CSS algorithm.

\section{Preliminaries}
\label{sec:preliminaries}

\subsection{System model}
\label{sec:system-model}

Here, we describe the system model we assumed in the paper.
The model definition follows that of SSS algorithm \cite{Moriya2001,Moriya2005}.
We consider distributed systems consisting of nodes that share no common (shared) memory or storage.
Nodes in the system can communicate with each other asynchronously, by exchanging messages (known as the message-passing model).
We assume that each node can send messages to any other node if the node knows the destination node's ID:
It can be realized if its underlying network supports appropriate multi-hop routing, even though the network is not completely connected.
Each node is a state machine and has a unique identifier (ID) drawn from a totally ordered set.
We assume a numerous but finite number of nodes can exist in the system.

We consider dynamic distributed systems, where nodes can frequently join and leave the distributed system.
This implies that the network topology of the system can change, and each node never recognizes the entire system's configurations in real time.
In our assumption, each node can join or leave the system freely, but to guarantee the consistency of the checkpoints, the node can leave the system only after taking a snapshot.
This implies that to leave, the node must initiate a snapshot algorithm.
If a message is sent to a node that has already left the system, the system informs the sender of the transmission failure.
On the other hand, a new coming node can join the system anytime.

Every (communication) link between nodes is reliable, which ensures that all the messages sent through the same link in the same direction are received, each exactly once, in the order they were sent (FIFO).
A message is received only when it is sent.
Because we assume an asynchronous distributed system, all messages are received in finite time (as long as the receiver exists), but with unpredictable delay.

\subsection{SSS algorithm}
\label{sec:sss-algorithm}

In this subsection, we briefly introduce SSS algorithm \cite{Moriya2001,Moriya2005} which takes a partial snapshot of a subsystem consisting of nodes communication-related to a single initiator.
This implies that SSS algorithm efficiently takes a partial snapshot; that is, the algorithm's message and time complexities do not depend on the total number of nodes in the distributed system.
SSS algorithm is also applicable to dynamic distributed systems, where nodes join and leave freely, because it does not require knowledge of the number of nodes or the topology of the system, but requires only a dynamically changing communication-relation among nodes.

In SSS algorithm, every node records its dependency set (\emph{DS}), which consists of the IDs of nodes with which it has communicated (sent or received messages).
SSS algorithm assumes that only a single node (called an \emph{initiator}) can initiate the algorithm, and to determine the subsystem, an initiator traces the communication-relation as follows:
When a node $p_i$ initiates SSS algorithm, the node records its current local state (as its checkpoint) and sends \mtag{Markers} with its ID to all nodes in its dependency set $DS_i$.
When a node $p_j$ receives a \mtag{Marker} message with the ID of $p_i$ for the first time, the node also records its current local state.
After that, $p_j$ forwards the \mtag{Markers} with the ID of $p_i$ to all nodes in its dependency set $DS_j$ and sends $DS_j$ to the initiator $p_i$.
The initiator can trace the communication-relation by referring the dependency sets received from other nodes:
The initiator maintains the union of the received dependency sets, including its own dependency set, and the set of the senders of the dependency sets.
When these two sets become the same, the nodes in the sets constitute the subsystem communication-related to the initiator.
The initiator informs each node $p_j$ in the determined subsystem of the node set of the subsystem; $p_j$ should receive \mtag{Markers} from every node in the set.

Recording in-transit messages in SSS algorithm is basically the same as in traditional distributed snapshot algorithms (Chandy and Lamport's manner).
Each node joining the partial snapshot algorithm records messages which are received before receipt of the \mtag{Marker} in each link.

\section{CPS Algorithm: The Proposed Algorithm}
\label{sec:proposed-algorithm}

\subsection{Overview}
When two or more nodes concurrently initiate SSS algorithm instances, the subsystems (called \emph{snapshot group}) may overlap, which is called \emph{a collision}.
CSS algorithm has been proposed with the aim of resolving this collision.
This algorithm combines the collided snapshot groups, using leader election repeatedly.
This allows concurrent initiations by two or more initiators; however, it causes a huge amount of communication cost for leader elections, if collisions occur frequently.
Moreover, to guarantee the consistency of the combined partial snapshot, every initiator must forward all information, e.g., the node list, the dependency set, and the collision-related information, to the leader.
This forwarding causes additional communication cost.

To reduce the communication cost, we propose a new partial snapshot algorithm, \emph{CPS algorithm}, which stands for \emph{Concurrent Partial Snapshot}.
This algorithm does not execute leader election to resolve a collision every time a collision is detected.
Instead, CPS algorithm creates a virtual link between the two initiators of the two collided groups, which is realized by making each initiator just store the other's ID as its neighbor's.
These links construct the overlay network which consists only of initiators.
We called this overlay network an \emph{initiator network}, and no information is forwarded among initiators in this network.
Figure \ref{fig:initiator-network} illustrates an example of an initiator network for a case where three snapshot groups collide with each other.

CPS algorithm consists of two phases: Concurrent Partial Snapshot Phase (Phase 1) and Termination Detection Phase (Phase 2).
In Phase 1, an initiator sends \mtag{Marker} messages to its communication-related nodes to determine its snapshot group.
If the snapshot group collides with another group, the initiator and the collided initiator create a virtual link between them for their initiator network.
When the snapshot group is determined, the initiator of the group proceeds to Phase 2 to guarantee the consistency of the checkpoints in all (overlapped) snapshot groups.
In Phase 2, to achieve the guarantee, each initiator communicates with each other in the initiator network to check all the initiators have already determined their snapshot groups.
After this check is completed, an initiator tells the termination condition of each node in the initiator's snapshot group and goes back to Phase 1 to finish the algorithm.
Note that all nodes in the snapshot groups execute Phase 1 on the real network, and only initiators execute Phase 2 on the initiator network that is constructed in Phase 1.

\begin{figure}[bt]
	\centering
	\includegraphics[scale=0.55]{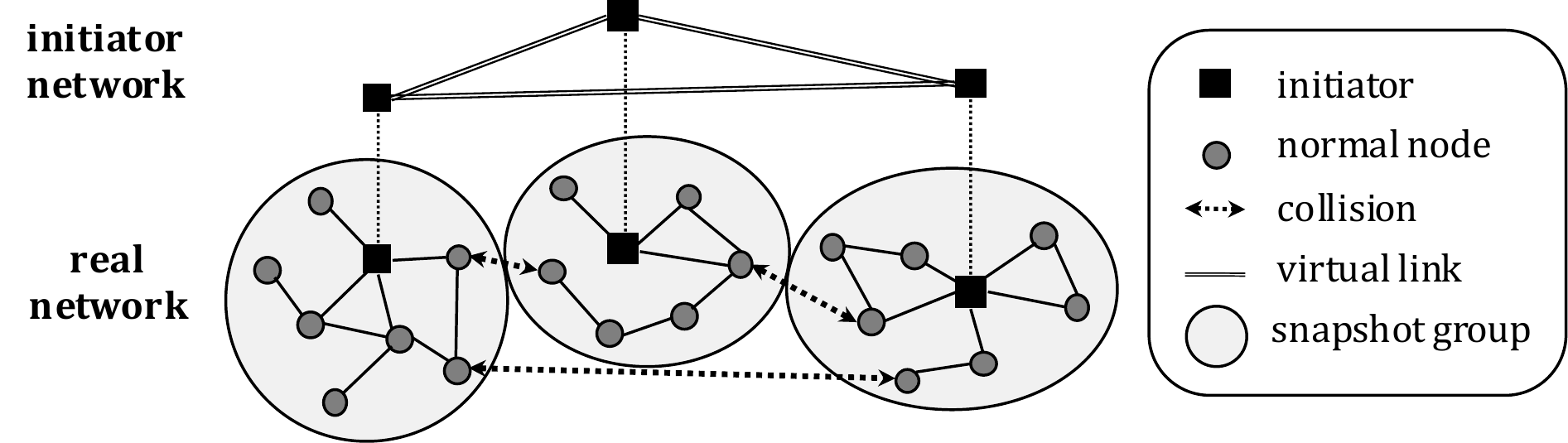}
	\caption{An (overlay) initiator network consisting of initiators}
	\label{fig:initiator-network}
\end{figure}

In this section, we describe the proposed CPS algorithm.
First, Section \ref{sec:basic-action} explains how the proposed algorithm handles events of sending/receiving an application message.
Then, Section \ref{sec:phase1-overview} and Section \ref{sec:phase2-overview} provide details of the two phases of the algorithm, i.e., Concurrent Partial Snapshot Phase and Termination Detection Phase.

\subsection{Basic operation}
\label{sec:basic-action}

To take a snapshot safely, CPS algorithm must handle events of sending or receiving an application message (as other snapshot algorithms do).
Algorithm \ref{alg:phase1basic} shows the operations that each node executes before sending (lines \ref{algl:send-msg-begin}--\ref{algl:send-msg-end}) or receiving (lines \ref{algl:receive-msg-begin}--\ref{algl:receive-msg-end}) an application message.
When node $p_i$ currently executing CPS algorithm ($init_i \neq \KwNull$) sends a message to node $p_j$ which is not in the $DS_i$, $p_i$ has to send \mtag{Marker} to $p_j$ before sending the message.
Variable $pDS$ stores DS when a node receives the first \mtag{Marker} to restore the content of DS when a snapshot algorithm is canceled.

Figure \ref{fig:orphan-example} depicts why this operation is necessary:
Let $p_k$ be the node which is communication-related to $p_i$ and $p_j$ ($p_i$ and $p_j$ are not communication-related with each other).
When each node receives \mtag{Marker} for the first time, the node broadcasts \mtag{Marker} to all the nodes in its $DS$.
Therefore, $p_i$ already sent \mtag{Marker} to $p_k$, and $p_k$ sends \mtag{Marker} to $p_j$ when these nodes receive the \mtag{Markers}.
However, if $p_i$ sends a message $m_{ij}$ to $p_j$ without sending \mtag{Marker} to $p_j$, the message might be received before the \mtag{Marker} from $p_k$, and it makes $m_{ij}$ an orphan message.
Let us consider another case in Fig.~\ref{fig:orphan-example} where $p_j$ sends $m_{ji}$ to $p_i$ before $p_j$ stores its checkpoint.
When $p_i$ receives $m_{ji}$, $p_i$ adds $m_{ji}$ into $MsgQ$ as defined in Algorithm \ref{alg:phase1basic} because $p_i$ is executing CPS algorithm and has not received a \mtag{Marker} message from $p_j$.
After finishing CPS algorithm, $m_{ji}$ is stored as one of the in-transit messages with the checkpoint.
Therefore, $m_{ji}$ never becomes an orphan message.

\begin{figure}[t]
	\centering
	\includegraphics[scale=0.55]{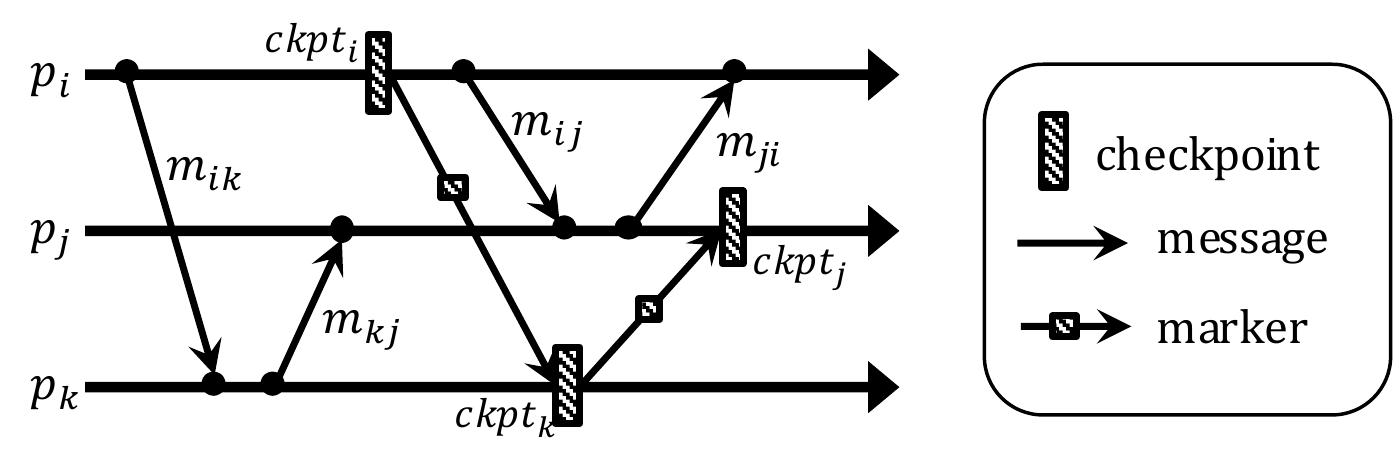}
	\caption{Orphan message $m_{ij}$}
	\label{fig:orphan-example}
\end{figure}

\begin{algorithm}[t]
\caption{Basic actions of Phase 1}
\label{alg:phase1basic}

\begin{algorithmic}[1]
\algfontsize
	\Procedure{Before $p_i$ sends a message to $p_j$}{} \label{algl:send-msg-begin}
	\If{$init \neq \KwNull \wedge p_j \notin pDS \cup DS \wedge InPhase2 = \KwFalse$}
		\State \Comment{Send \mtag{Marker} before sending a message}
		\State{Send \msg{\mtag{Marker}, init} to $p_j$}
	\EndIf
	\State{$DS \leftarrow DS \cup \{p_j\}$} \Comment{Add $p_j$ to its $DS$}
\EndProcedure \label{algl:send-msg-end}

\Procedure{Before $p_i$ receives a message from $p_j$}{} \label{algl:receive-msg-begin}
	\State{$DS \leftarrow DS \cup \{p_j\}$}\Comment{Add $p_j$ to its $DS$}
	\If{$init \neq \KwNull \wedge p_j \notin RcvMk$}
		\State{Add ($p_j$, message) to $MsgQ$}
	\EndIf
\EndProcedure \label{algl:receive-msg-end}

\end{algorithmic}
\end{algorithm}

\subsection{Phase 1: Concurrent Partial Snapshot Phase}
\label{sec:phase1-overview}

This phase is basically the same as that in SSS algorithm, except for the collision-handling process.
Each node can initiate a snapshot algorithm at any time, by sending a special message \emph{Marker} to the node's communication-related nodes, and the other nodes record their local states when they receive \mtag{Marker} for the first time.
An initiator of CPS algorithm traces the communication-relation to determine its partial snapshot group.

In Phase 1, each node $p_i$ maintains the following variables:
\begin{itemize}
	\setlength{\itemsep}{0em}
    \item $init_i$: Initiator's ID.
		An initiator sets this variable as its own ID.
		A normal node (not initiator) sets this variable to the initiator ID of the first \mtag{Marker} message it receives.
		Initially \KwNull.
	\item $DS_i$: A set of the IDs of the (directly) communicate-related nodes.
		This set is updated when $p_i$ sends/receives an application message as described in Section \ref{sec:basic-action}.
	\item $pDS_i$: A set variable that stores the $DS_i$ temporarily.
		Initially $\emptyset$.
	\item $RcvMk_i$: A set of the IDs of the nodes from which $p_i$ (already) received \mtag{Marker} messages.
		Initially $\emptyset$.
	\item $MkList_i$: A set of the IDs of the nodes from which $p_i$ has to receive \mtag{Marker} messages to terminate the algorithm.
		Initially $\emptyset$.
	\item $fin_i$: A boolean variable that denotes whether the partial snapshot group is determined or not.
		Initially \KwFalse.
		An initiator updates this variable to {\KwTrue} when Phase 1 terminates, while a non-initiator node updates this when the node receives a \mtag{Fin} message.
	\item $MsgQ_i$: A message queue that stores a sequence of the messages for checkpoints, as the pairs of the ID of the sender node and the message.
		Initially \KwNull.
	\item $CollidedNodes_i$: A set of the IDs of the nodes from which $p_i$ received collision \mtag{Marker} messages.
		Initially $\emptyset$.
	\item $MkFrom_i$ (Initiator only): A set of the IDs of the nodes that send \mtag{Marker} to its DS.
		Initially $\emptyset$.
	\item $MkTo_i$ (Initiator only): The union set of the DSes of the nodes in $MkFrom$.
		Initially $\emptyset$.
	\item $DSInfo_i$ (Initiator only): A set of the pairs of a node ID and its DS.
		Initially $\emptyset$.
	\item $Wait_i$ (Initiator only): A set of the IDs of the nodes from which $p_i$ is waiting for a reply to create a virtual link of the initiator network.
		Initially $\emptyset$.
	\item $N_i$ (Initiator only): A set of the neighbor nodes' IDs in the initiator network.
		Initially $\emptyset$.
\end{itemize}

We use the following message types in Phase 1.
We denote the algorithm messages by \msg{\mtag{MessageType}, arg_1, arg_2, \ldots}.
Note that some messages have no argument.
We assume that every message includes the sender ID and the snapshot instance ID, which is a pair of an initiator ID and a sequence number of the snapshot instances the initiator invoked, to distinguish snapshot algorithm instances that are or were executed.
\begin{itemize}
	\setlength{\itemsep}{0em}
	\item \msg{\mtag{Marker}, init}: A message which controls the timing of the recording of the local state.
		Parameter $init$ denotes the initiator's ID.
	\item \msg{\mtag{MyDS}, DS}: A message to send its own DS (all nodes communication-related to this node) to its initiator.
	\item \msg{\mtag{Out}}: A message to cancel the current snapshot algorithm.
		When a node who has been an initiator receives a \mtag{MyDS} message of the node's previous instance, the node sends this message to cancel the sender's snapshot algorithm instance.
	\item \msg{\mtag{Fin}, List}: A message to inform that its partial snapshot group is determined.
		$List$ consists of the IDs of the nodes from which the node has to receive \mtag{Marker} messages to terminate the algorithm.
	\item \msg{\mtag{NewInit}, p, Init}: A message to inform that a different initiator has been detected.
		$Init$ denotes the ID of the detected initiator, and $p$ denotes the ID of the node which sends \mtag{Marker} with $Init$.
	\item \msg{\mtag{Link}, p, q}: A message sent by an initiator to another initiator to confirm whether a link (of the overlay network) can be created between the two initiators or not.
		$p$ denotes the ID of the node which received a collided \mtag{Marker}, and $q$ denotes the ID of the sender node.
	\item \msg{\mtag{Ack}, p, q}: A reply message for a \msg{Link, p, q} message when the link can be created.
	\item \msg{\mtag{Deny}, p, q}: A reply message for a \msg{Link, p, q} message when the link cannot be created.
	\item \msg{\mtag{Accept}, p, Init}: A reply message for a \msg{NewInit, p, Init} message when the link between its initiator and $Init$ is successfully created.
\end{itemize}

Algorithm \ref{alg:phase1-normaloperations} presents the pseudo-code of Phase 1.
By this algorithm, each node stores, as a checkpoint, a local application state in line \ref{algl:store-localstate} and in-transit messages in line \ref{algl:store-msgq}.

We briefly present how an initiator determines its partial snapshot group when no collision occurs.
Figure \ref{fig:psg-example} describes an example of a distributed system consisting of 10 nodes, $p_0$ to $p_9$, and some pairs are communication-related:
For example, $p_7$ has communication-relations with $p_0$, $p_6$, and $p_8$; i.e., $DS_7 = \{p_0, p_6, p_8\}$.
In this example, $p_0$ initiates CPS algorithm.
$p_0$ initializes all variables, and records its local state; then, $p_0$ sends \msg{\mtag{Marker}, p_0} to all nodes in $DS_0 = \{p_2, p_3, p_6, p_7\}$ (lines \ref{algl:p1-case-a-begin}--\ref{algl:p1-case-a-end}).
When $p_3$ receives the first \mtag{Marker} from $p_0$, $p_3$ records its local state, and sets $p_0$ as its initiator (variable $init_3$) (lines \ref{algl:p1-case-a-begin}--\ref{algl:store-localstate}).
Then, $p_3$ sends its $DS_3$ to its initiator $p_0$ using the \msg{\mtag{MyDS}, DS_3} message (line \ref{algl:send-myds}).
After that, $p_3$ sends \msg{\mtag{Marker}, p_0} to all nodes in $DS_3 = \{p_0, p_8\} (line \ref{algl:p1-case-a-end})$.
Note that node $p_8$, which is not directly communication-related to $p_0$, also receives \msg{\mtag{Marker}, p_0} from $p_3$ (or $p_7$) and records its local state.
If the initiator $p_0$ receives a \msg{\mtag{MyDS}, DS_i} message from $p_i$, it adds the ID $p_i$ and $DS_i$ to $MkFrom_0$ and $MkTo_0$ respectively, and inserts $(i, DS_i)$ into $DSInfo_0$ (lines \ref{algl:update-myds-begin}--\ref{algl:update-myds-end}).
When $MkTo_0 \subseteq MkFrom_0$\footnote{If $DS_0$ remains unchanged, $MkTo_0 = MkFrom_0$ holds.
However, each node $p_i$ can send a message to a node not in $DS_i$ (which adds the node to $DS_i$) even while CPS algorithm is being executed.
This may cause $MkTo_0 \subset MkFrom_0$; refer to Algorithm \ref{alg:phase1basic} for details.} holds, this means that all nodes which are communication-related to the initiator already received the \mtag{Marker}.
Thus, the initiator determines its partial snapshot group as the nodes in $MkFrom_0$, and proceeds to Phase 2 (lines \ref{algl:determine-psg-begin}--\ref{algl:determine-psg-end}), named the \emph{Termination Detection Phase}, which is presented in the next subsection.
When Phase 2 finishes, the initiator sends the \msg{\mtag{Fin}, MkList_i} message to each $p_i \in MkFrom_0$ (lines \ref{algl:p2-fin-begin}--\ref{algl:p2-fin-end} of Algorithm \ref{alg:phase2}), where $MkList_i$ is the set of the IDs from which $p_i$ has to receive \mtag{Markers}.
If node $p_i$ has received \mtag{Marker} messages from all the nodes in $MkList_i$, $p_i$ terminates the algorithm (lines \ref{algl:checkterm-begin}--\ref{algl:checkterm-end}).

\begin{figure}[t]
	\centering
	\includegraphics[scale=0.55]{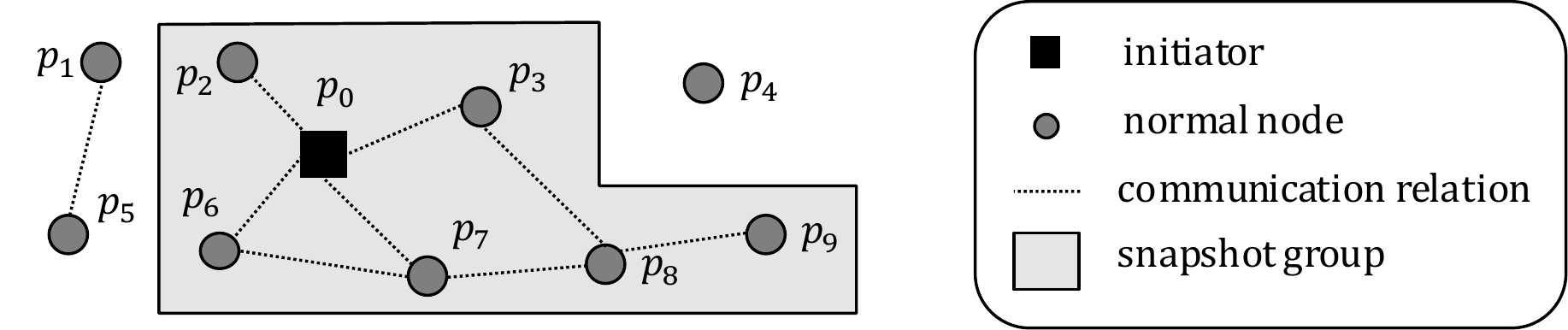}
	\caption{Partial snapshot group example}
	\label{fig:psg-example}
\end{figure}

\begin{algorithm}[t]
\caption{Pseudo code of CPS algorithm for node $p_i$ (normal operations of Phase 1)}
\label{alg:phase1-normaloperations}
\begin{algorithmic}[1]
\algfontsize
\Procedure{Initiate}{ }
	\State OnReceive(\msg{\mtag{Marker}, p_i}) \label{algl:initiation}
\EndProcedure

\Statex
\Procedure{OnReceive}{\msg{\mtag{Marker}, p_x} from $p_j$} \label{algl:marker-begin}
	\If{$init_i = \KwNull$}
		\State \Comment{This is the first \mtag{Marker}} \label{algl:p1-case-a-begin}
		\State $init_i \gets p_x$, $RcvMk_i \gets RcvMk_i \cup \{p_j\}$
		\State $pDS_i \gets DS_i$, $DS_i \gets \emptyset$
		\State $MkList_i \gets \emptyset$, $fin_i \gets \KwFalse$
		\State $MsgQ_i \gets \emptyset$
		\State Record its own local state \label{algl:store-localstate}
		\State Send \msg{\mtag{MyDS}, pDS_i} to $init_i$ \label{algl:send-myds}
		\State Send \msg{\mtag{Marker}, p_x} to $\forall p_k \in pDS_i$ \label{algl:p1-case-a-end}

	\ElsIf{$init_i = p_x$}
		\State \Comment{\mtag{Marker} from the same snapshot group} \label{algl:p1-case-b-begin}
		\State $RcvMk_i \gets RcvMk_i \cup \{p_j\}$
		\If{$fin_i = \KwTrue$}
			\State $\proc{CheckTermination}()$
		\EndIf \label{algl:p1-case-b-end}

	\ElsIf{$init_i \neq p_x$}
		\State \Comment{A collision occurs} \label{algl:p1-case-c-begin}
		\State $RcvMk_i \gets RcvMk_i \cup \{p_j\}$
		\State $CollidedNodes_i \gets CollidedNodes_i \cup \{(p_j, p_x)\}$
		\If{$fin_i = \KwFalse$}
			\State Send \msg{\mtag{NewInit}, p_j, p_x} to $init_i$ \label{algl:send-newinit}
		\EndIf \label{algl:p1-case-c-end}

	\EndIf

\EndProcedure \label{algl:marker-end}

\Statex
\Procedure{OnReceive}{\msg{\mtag{MyDS}, DS_j} from $p_j$}
	\If{$init_i = \KwNull \vee fin_i = \KwTrue$}
		\State Send \msg{\mtag{Out}} to $p_j$
	\Else
		\State $MkFrom_i \gets MkFrom_i \cup \{p_j\}$ \label{algl:update-myds-begin}
		\State $MkTo_i \gets MkTo_i \cup \{DS_j\}$
		\State $DSInfo_i \gets DSInfo_i \cup (p_j, DS_j)$ \label{algl:update-myds-end}
		\State $\proc{CanDetermineSG}()$
	\EndIf
\EndProcedure

\Statex
\Procedure{OnReceive}{\msg{\mtag{Out}} from $p_j$}
	\State \Comment{Cancel its snapshot algorithm}
	\State $init_i \gets \KwNull$
	\State $DS_i \gets DS_i \cup pDS_i$
	\State Delete recorded local state and received messages in $MsgQ_i$
	\State $\proc{ReProcessMarker}()$
\EndProcedure

\Statex
\Procedure{OnReceive}{\msg{\mtag{Fin}, List} from $p_j$}
	\State $MkList_i \gets List$
	\State \Comment{My initiator notifies the determination of its snapshot group}
	\State $fin_i \gets \KwTrue$
	\State $\proc{CheckTermination}()$
\EndProcedure

\Statex
\Procedure{OnTermination}{ }
	\State $\proc{ReProcessMarker}()$
\EndProcedure

\algstore{normaloperations}
\end{algorithmic}
\end{algorithm}

\setcounter{algorithm}{1}
\begin{algorithm}[t]
\caption{Pseudo code of CPS algorithm for node $p_i$ (normal operations of Phase 1) (Cont'd)}
\begin{algorithmic}[1]
\algfontsize
\algrestore{normaloperations}

\Procedure{CanDetermineSG()}{} \label{algl:candeterminesg-begin}
	\If{$MkTo_i \subseteq MkFrom_i \wedge Wait_i = \emptyset$} \label{algl:candeterminesg-condition}
		\State \Comment{Initiator $p_i$ determines its snapshot group} \label{algl:determine-psg-begin}
		\State $fin_i \gets \KwTrue$
		\State $\proc{StartPhase2}()$ \label{algl:determine-psg-end}
	\EndIf
\EndProcedure \label{algl:candeterminesg-end}

\Statex
\Procedure{CheckTermination()}{} \label{algl:checkterm-begin}
	\If{$MkList_i \subseteq RcvMk_i$}
		\For{each $(p_j, m)$ in $MsgQ_i$}
			\If{$p_j \in MkList_i$}
				\State Record $m$ as an in-transit message \label{algl:store-msgq}
			\EndIf
		\EndFor
		\State Wait until $InPhase2_i = \KwFalse$
		\State Terminate this snapshot algorithm
	\EndIf
\EndProcedure \label{algl:checkterm-end}

\Statex
\Procedure{ReProcessMarker}{ }
	\If{$CollidedNodes_i \neq \emptyset$}
		\State \Comment{Process \mtag{Markers} again for collisions that is not resolved}
		\For{each $(p_y, p_b) \in CollidedNodes_i$}
			\State \proc{OnReceive}(\msg{\mtag{Marker}, p_b} from $p_y$)
		\EndFor
	\EndIf
\EndProcedure

\end{algorithmic}
\end{algorithm}

Algorithm \ref{alg:collision-handling} presents the pseudo-code of the collision-handling procedures in Phase 1.
In the algorithm, we change some notations of node IDs for ease of understanding.
Our assumption is depicted in Figure \ref{fig:collision-notation}.
We assume that a collision occurs between two snapshot groups, and let $p_x$ and $p_y$ be the nodes executing the snapshot algorithm by receiving \mtag{Marker} from the initiators $p_a$ and $p_b$, respectively.
Node $p_x$ receives \msg{\mtag{Marker}, p_b} from $p_y$, and $p_x$ informs its initiator $p_a$ of a collision by sending a \mtag{NewInit} message, because $init_x \neq p_b$.

\begin{figure}[t]
	\centering
	\includegraphics[scale=0.55]{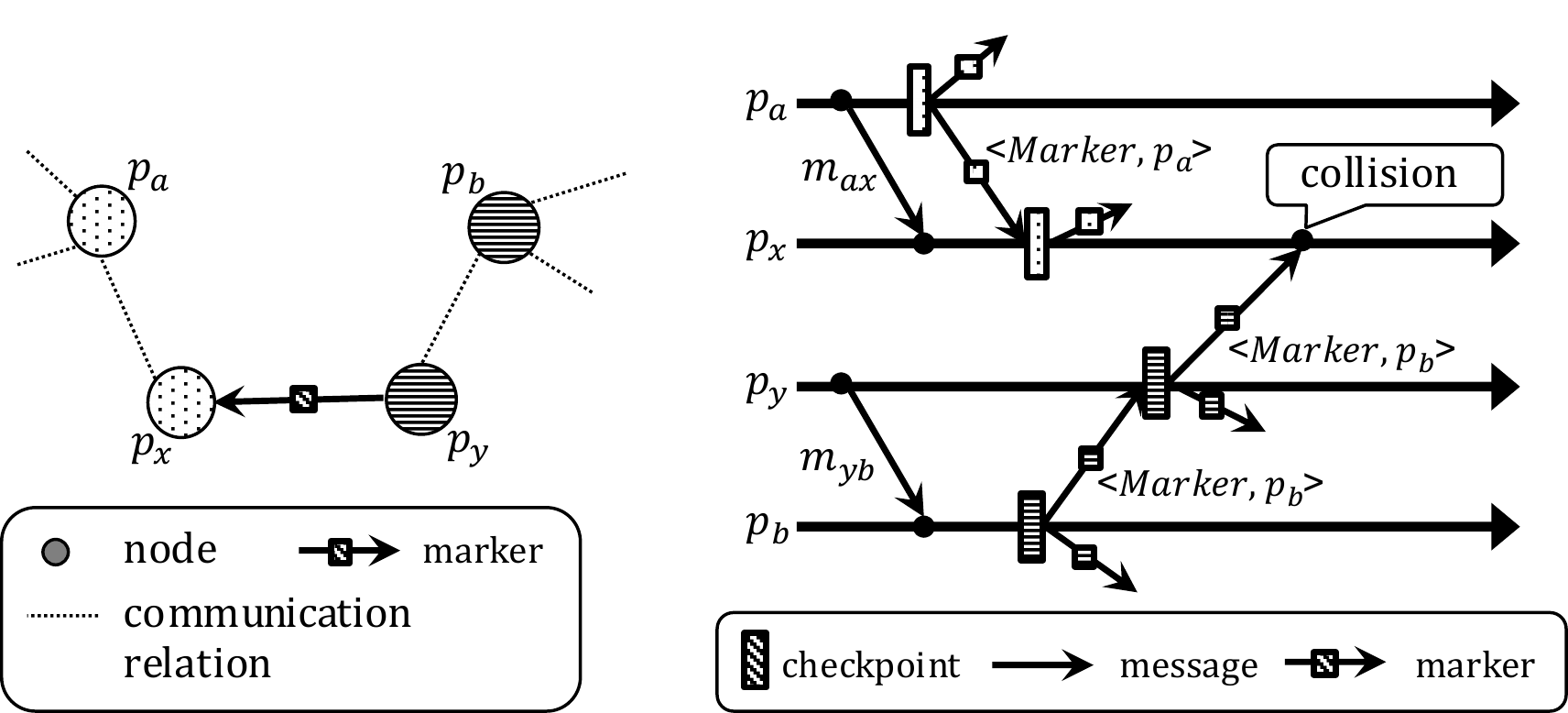}
	\caption{Collision assumption of Algorithm \ref{alg:collision-handling}}
	\label{fig:collision-notation}
\end{figure}

\begin{algorithm}[t]
\caption{Pseudo code of CPS algorithm (collision handling of Phase 1)}
\label{alg:collision-handling}
\begin{algorithmic}[1]
\algfontsize

\State \Comment{From the view of $p_a$ in Fig.~\ref{fig:collision-notation}}
\Procedure{OnReceive}{\msg{\mtag{NewInit}, p_y, p_b} from $p_x$}
	\If{$fin_a = \KwFalse$}
		\If{$p_b \notin N_a$}
			\State $Wait_a \gets Wait_a \cup (p_x, p_y, p_b)$ \label{algl:wait-add}
			\State Send \msg{\mtag{Link}, p_x, p_y} to $p_b$ \label{algl:send-link}
		\Else
			\State $MkFrom_a \gets MkFrom_a \cup \{p_y\}$
			\State $MkTo_a \gets MkTo_a \cup \{p_x\}$
			\State $DSInfo_a \gets DSInfo_a \cup (p_y, \{p_x\})$
			\State Send \msg{\mtag{Link}, p_x, p_y} to $p_b$
			\State Send \msg{\mtag{Accept}, p_y, p_b} to $p_x$
		\EndIf
	\ElsIf{$p_b \in N_a$}
		\State Send \msg{\mtag{Link}, p_x, p_y} to $p_b$
	\EndIf
\EndProcedure

\Statex
\State \Comment{From the view of $p_b$ in Fig.~\ref{fig:collision-notation}}
\Procedure{OnReceive}{\msg{\mtag{Link}, p_x, p_y} from $p_a$} \label{algl:link-begin}
	\If{$fin_b = \KwFalse$}
		\State $MkFrom_b \gets MkFrom_b \cup \{p_x\}$ \label{algl:link-mkfrom-update}
	    \If{$p_a \notin N_b$}
    		\State $N_b \gets N_b \cup \{p_a\}$
    		\State $MkTo_b \gets MkTo_b \cup \{p_y\}$
    		\State $DSInfo_b \gets DSInfo_b \cup (p_x, \{p_y\})$
			\State Send \msg{\mtag{Ack}, p_x, p_y} to $p_a$ \label{algl:send-ack}
    		\State $\proc{AcceptColliededNodes}(p_a)$
    		\State $\proc{CanDetermineSG}()$
    	\EndIf
	\Else
		\State Send \msg{\mtag{Deny}, p_x, p_y} to $p_a$ \label{algl:send-deny}
	\EndIf
\EndProcedure \label{algl:link-end}

\State \Comment{From the view of $p_a$ in Fig.~\ref{fig:collision-notation}}
\Procedure{OnReceive}{\msg{\mtag{Ack}, p_x, p_y} from $p_b$}
	\State $N_a \gets N_a \cup \{p_b\}$
	\State $\proc{AcceptColliededNodes}(p_b)$
	\State $\proc{CanDetermineSG}()$
\EndProcedure

\Statex
\State \Comment{From the view of $p_a$ in Fig.~\ref{fig:collision-notation}}
\Procedure{OnReceive}{\msg{\mtag{Deny}, p_x, p_y} from $p_b$}
	\State $Wait_a \gets Wait_a \setminus \{(p_x, p_y, p_b)\}$
	\If{$p_b \notin N_a$}
		\State $\proc{CanDetermineSG}()$
	\EndIf
\EndProcedure

\Statex
\State \Comment{From the view of $p_x$ in Fig.~\ref{fig:collision-notation}}
\Procedure{OnReceive}{\msg{\mtag{Accept}, p_y, p_b} from $p_a$}
	\If{$p_y \notin pDS_x$}
		\State Send \msg{\mtag{Marker}, p_b} to $p_y$ \label{algl:send-marker-for-collision}
	\EndIf
	\State $CollidedNodes_x \gets CollidedNodes_x \setminus \{(p_y, p_b)\}$
\EndProcedure

\Statex
\State \Comment{From the view of $p_a$ in Fig.~\ref{fig:collision-notation}}
\Procedure{AcceptCollidedNodes}{$p_b$}
	\For{each $(p_i, p_j, p_b) \in Wait$}
		\State $MkFrom \gets MkFrom \cup \{p_j\}$
		\State $MkTo \gets MkTo \cup \{p_i\}$
		\State $DSInfo \gets DSInfo \cup (p_j, \{p_i\})$
		\State Send \msg{\mtag{Accept}, p_j, p_k} to $p_i$ \label{algl:send-accept}
		\State $Wait \gets Wait \setminus \{(p_i, p_j, p_k)\}$
	\EndFor
\EndProcedure

\end{algorithmic}
\end{algorithm}

Figure \ref{fig:collision-example} illustrates an example of the message flow when a collision occurs.
In the example, we assume that two initiators, $p_0$ and $p_6$, concurrently initiate CPS algorithm instances, and $p_4$ detects a collision as follows.
Node $p_4$ receives \msg{\mtag{Marker}, p_0} from $p_3$, and \msg{\mtag{Marker}, p_6} from $p_5$ in this order.
Because $p_4$ receives \mtag{Marker} with initiator $p_6$ different from its initiator $p_0$, $p_4$ sends \msg{\mtag{NewInit}, p_5, p_6} to its initiator $p_0$ (line \ref{algl:send-newinit} of Algorithm \ref{alg:phase1-normaloperations}).
When $p_0$ receives the \mtag{NewInit}, if $p_0$ has not determined the partial snapshot group yet, $p_0$ sends a \msg{\mtag{Link}, p_4, p_5} message to opponent initiator $p_6$ (line \ref{algl:send-link}).
As a reply to the \mtag{Link} message, $p_6$ sends a \msg{\mtag{Ack}, p_4, p_5} message (line \ref{algl:send-ack}), if $p_6$ also has not determined its partial snapshot group yet.
Otherwise, $p_6$ sends a \msg{\mtag{Deny}, p_4, p_5} message to $p_0$\footnote{In the \mtag{Deny} case, $p_6$ has determined its snapshot group and has sent \mtag{Fin} messages to the nodes in the group including $p_5$.
Node $p_5$ eventually receives the \mtag{Fin} message and terminates the snapshot algorithm.
While $p_4$ cannot receive any response for the \mtag{NewInit} message $p_4$ sent, the node also eventually receives its \mtag{Fin} message from $p_0$.
If there exists an application message $m_{54}$ sent from $p_5$ to $p_4$, the sent must be after taking the checkpoint of $p_5$ for $p_6$'s snapshot (otherwise, the two snapshot groups of $p_0$ and $p_6$ are merged).
Node $p_4$ also receives $m_{54}$ after taking its checkpoint for $p_0$'s snapshot.
If $p_4$ receives the message before its checkpoint, $p_5$ send a \mtag{Marker} to $p_4$ before $m_{54}$, $p_4$ should join in $p_6$'s snapshot group.
The application message $m_{54}$ is sent and received after the checkpoints of $p_4$ and $p_5$; thus, the message never becomes an orphan.
We can have the same discussion for an application message sent in the opposite direction.
} (line \ref{algl:send-deny}).
Finally, $p_0$ sends \msg{\mtag{Accept}, p_5, p_6} to $p_4$ which detected the collision (line \ref{algl:send-accept}), and $p_4$ sends \msg{\mtag{Marker}, p_6} to $p_5$ (line \ref{algl:send-marker-for-collision}).
Note that this \mtag{Marker} is necessary to decide which messages should be recorded in the checkpoint in $p_5$.
In this example, we also notice the following points: (1) In Figure \ref{fig:collision-example}, $p_5$ may also detect a collision by \msg{\mtag{Marker}, p_0} from $p_4$.
This causes additional message communications between $p_0$ and $p_6$; e.g., $p_6$ also sends a \mtag{Link} message to $p_0$.
(2) Even if there is no communication-relation between $p_4$ and $p_5$, when two initiators invoke CPS algorithm instances, $p_4$ or $p_5$ can send \mtag{Marker} in advance to send a message (refer to Algorithm \ref{alg:phase1basic}).
In this case, a virtual link between $p_0$ and $p_6$ may not be created, because either of them may have already determined their partial snapshot groups (note that $p_5$ and $p_4$ are not included in $DS_4$ and $DS_5$, respectively).

\begin{figure}[t]
	\centering
	\includegraphics[scale=0.55]{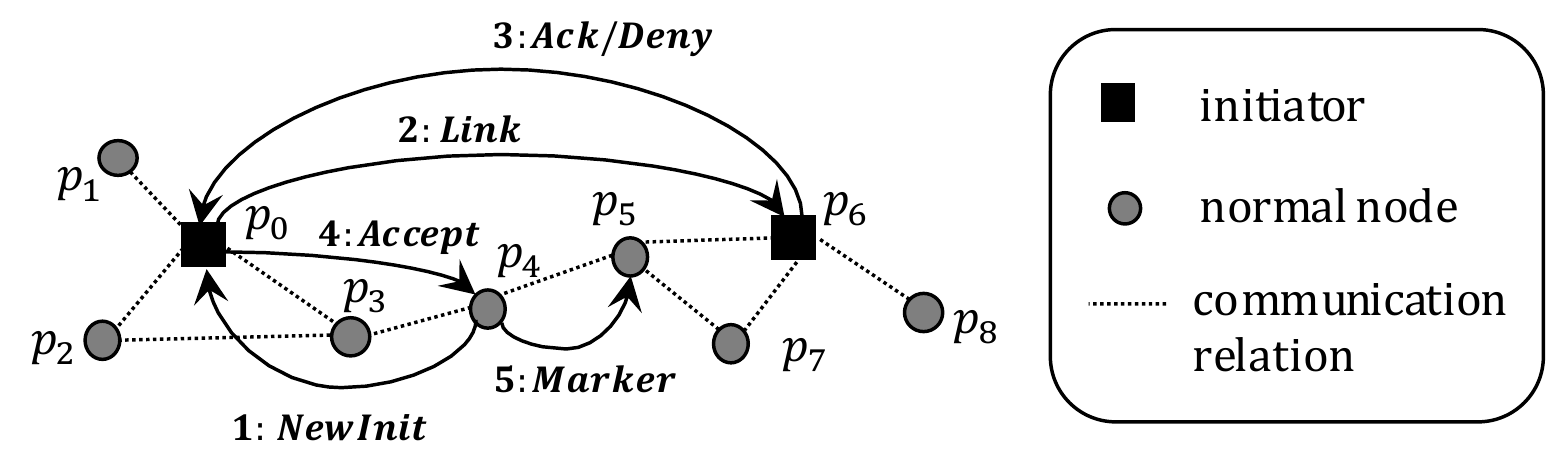}
	\caption{Collision-handling example in CPS algorithm}
	\label{fig:collision-example}
\end{figure}

\subsection{Phase 2: Termination Detection Phase}
\label{sec:phase2-overview}

Only the initiators, which determine their partial snapshot groups, execute Phase 2.
Note that Phase 2 is executed on the initiator network that was constructed in Phase 1.
The goal of this phase is to confirm that all initiators in the initiator network have already determined their snapshot groups\footnote{If an initiator has not experienced any collision in Phase 1, the initiator terminates Phase 2 immediately because the initiator does not need to wait other snapshot groups.}.
In other words, all initiators in the initiator network completed Phase 1, and are executing Phase 2.
In this phase, the proposed algorithm elects one initiator as the leader, and constructs a breadth-first spanning tree rooted at the leader.
From the leaves to the root, each initiator notifies its parent initiator in the tree that it is in Phase 2 (\emph{convergecast}), and when the convergecast terminates, the leader broadcasts the termination of Phase 2 to all other initiators (\emph{broadcast}).

In Phase 2, each initiator $p_i$ maintains the following variables:
\begin{itemize}
	\setlength{\itemsep}{0em}
	\item $rID_i$: The ID of the root initiator the initiator currently knows.
		Initially, \KwNull.
	\item $dist_i$: The distance to the root initiator $rID_i$.
		Initially, $\infty$.
	\item $pID_i$: The ID of the parent initiator in the (spanning) tree rooted at the root initiator $rID_i$.
		Initially, \KwNull.
	\item $Child_i$: A set of the IDs of the child initiators in the (spanning) tree.
		Initially, $\emptyset$.
	\item $LT_i$: A set of the IDs of the initiator from which the initiator received \mtag{LocalTerm} messages.
		Initially, $\emptyset$.
	\item $CK_i$: A set of the IDs of the initiator from which the initiator received \mtag{Check} messages.
		Initially, $\emptyset$.
	\item $InPhase2_i$: A boolean variable.
		This is {\KwTrue} if $p_i$ is in Phase 2; otherwise, {\KwFalse}.
\end{itemize}
In addition, the following Phase 1 variables are also used.
Note that these variables are never updated in Phase 2.
\begin{itemize}
	\setlength{\itemsep}{0em}
	\item $MKFrom_i$
	\item $DSInfo_i$
\end{itemize}

The following messages are used in Phase 2.
\begin{itemize}
	\setlength{\itemsep}{0em}
	\item \msg{\mtag{Check}, rID, dist, pID}: A message to inform its neighbors of the smallest ID that the initiator currently knows.
		$rID$ is the initiator that has the smallest ID (the initiator currently knows), $dist$ is the distance to $rID$, and $pID$ is the parent initiator's ID to $rID$.
	\item \msg{\mtag{LocalTerm}}: A message for a convergecast.
	\item \msg{\mtag{GlobalTerm}}: The leader initiator (which has the smallest ID) broadcasts this message to all other initiators when a convergecast is successfully finished.
\end{itemize}

\begin{algorithm}[t]
\caption{Pseudo code of CPS algorithm for initiator $p_i$ (Phase 2)}
\label{alg:phase2}
\begin{algorithmic}[1]
\algfontsize

\Procedure{StartPhase2()}{}
	\If{$N_i \neq \emptyset$}
		\State $rID_i \gets p_i$, $dist_i \gets 0$, $pID_i \gets p_i$, $Child_i \gets \emptyset$
		\State $LT_i \gets \emptyset$, $CK_i \gets \emptyset$, $InPhase2_i \gets \KwTrue$
		\State Send \msg{\mtag{Check}, rID_i, dist_i, pID_i} to $\forall p_j \in N_i$ \label{algl:send-check-first}
		\State Process the messages arrived before entering Phase 2
	\Else
		\State \Comment{There are no neighbors on the initiator network}
		\State $\proc{FinishPhase2}()$
	\EndIf
\EndProcedure

\Statex
\Procedure{OnReceive}{\msg{\mtag{Check}, rID_j, dist_j, pID_j} from $p_j \in N_i$}
	\State $CK_i \gets CK_i \cup \{p_j\}$
	\If{$rID_j < rID_i \vee (rID_j = rID_i \wedge dist_j+1 < dist_i)$}
		\State $rID_i \gets rID_j$, $dist_i \gets dist_j + 1$, $pID_i \gets p_j$ \label{algl:update-root}
		\State Send \msg{\mtag{Check}, rID_i, dist_i, pID_i} to $\forall p_j \in N_i$ \label{algl:send-check}
	\EndIf

	\If{$pID_j = p_i$}
		\State $Child_i \gets Child_i \cup \{p_j\}$
	\ElsIf{$p_j \in Child_i$}
		\State $Child_i \gets Child_i \setminus \{p_j\}$
		\State $LT_i \gets LT_i \setminus \{p_j\}$
	\EndIf
	\If{$CK_i = N_i \wedge Child_i = \emptyset$}
		\State Send \msg{\mtag{LocalTerm}} to $pID_i$ \label{algl:send-lt}
	\EndIf
\EndProcedure

\Statex
\Procedure{OnReceive}{\msg{\mtag{LocalTerm}} from $p_j \in N_i$}
		\State $LT_i \gets LT_i \cup \{p_j\}$ \label{algl:update-lt}
	\If{$Child_i = CK_i = LT_i = N_i \wedge pID_i = p_i$}
		\State Send \msg{\mtag{GlobalTerm}} to $\forall p_j \in Child_i$ \label{algl:send-gt}
		\State $\proc{FinishPhase2}()$
	\ElsIf{$Child_i = LT_i \wedge CK_i = N_i$}
		\State Send \msg{\mtag{LocalTerm}} to $pID_i$
	\EndIf
\EndProcedure

\Statex
\Procedure{OnReceive}{\msg{\mtag{GlobalTerm}} from $p_j \in N_i$}
	\State Send \msg{\mtag{GlobalTerm}} to $\forall p_j \in Child_i$
	\State $\proc{FinishPhase2}()$
\EndProcedure

\Statex
\Procedure{FinishPhase2()}{}
	\State $InPhase2_i \gets \KwFalse$
	\For{each $p_k \in MkFrom_i$} \label{algl:p2-fin-begin}
		\State $MkList_k \gets \{ \forall p_x \mid p_k \in DS_x, (p_x, DS_x) \in DSInfo_i\}$
		\State Send \msg{\mtag{Fin}, MkList_k} to $p_k$
	\EndFor \label{algl:p2-fin-end}
\EndProcedure

\end{algorithmic}
\end{algorithm}

Algorithm \ref{alg:phase2} presents the pseudo-code of the proposed algorithm of Phase 2.
In Phase 2, each initiator repeatedly broadcasts a \mtag{Check} message to its neighbor initiators, to find the leader.
The \mtag{Check} message includes the smallest ID (denoted by $rID$) that the initiator ever knows and the distance to it.
When an initiator starts Phase 2, the initiator sends a \mtag{Check} message containing its ID as the minimum ID $rID$ to its all neighbor initiators (line \ref{algl:send-check-first}).
When the initiator receives \mtag{Check} messages, it updates its root, its distance, and its parent initiator (line \ref{algl:update-root}), if it finds a smaller ID or a smaller distance with the smallest ID it ever knows.
If there is some update on these variables, it sends the \mtag{Check} message with the updated information to all its neighbor initiators again (line \ref{algl:send-check}).
By repeating these broadcasts and updates, initiators construct a breadth-first spanning tree rooted at the initiator with the smallest ID.

This naive technique is widely used to find the leader in the distributed system.
However, this technique is hardly applicable when the diameter of the network is unknown, because the broadcast of the \mtag{Check} message has to be repeated as many times as the diameter of the network.
To resolve this difficulty, in the proposed algorithm, we allow an initiator $p_i$ to stop broadcasting \mtag{Check} and start convergecast toward the leader (the initiator currently knows), when the following conditions are satisfied (line \ref{algl:send-lt}):
(1) an initiator $p_i$ receives \mtag{Check} messages from its all neighbor initiators, and (2) there are no child initiators in the neighbors.
This implies that initiator $p_i$ is a leaf initiator of the tree rooted at the leader.
Even after an initiator begins the convergecast, the initiator stops it when the initiator receives a \mtag{Check} message from any neighbor initiator, and the initiator restarts the convergecast when the conditions above are satisfied.

The convergecast uses a \mtag{LocalTerm} message that is repeatedly sent from a leaf initiator to the root initiator (the leader) through the tree.
When the initiator receives a \mtag{LocalTerm} message, the initiator adds the sender's ID to its set variable $LT$ (line \ref{algl:update-lt}), which is a set variable that stores the IDs of the initiators from which the initiator received \mtag{LocalTerm} messages.
Therefore, the parent initiator (which has one or more child initiators) starts the convergecast when the initiator receives \mtag{LocalTerm} messages from all its child initiators (line \ref{algl:send-lt}).
The convergecast is terminated when the leader receives \mtag{LocalTerm} messages from all its neighbor initiators (note that all neighbor initiators of the leader eventually become the leader's children), and the leader broadcast \mtag{GlobalTerm} messages to finish Phase 2 (line \ref{algl:send-gt}).
This implies that to terminate the convergecast, all initiators have to start convergecasts, and this means all initiators have the same $rID$.
If some initiators start convergecasts with wrong information, e.g., the rID of the initiator is not the smallest ID, these initiators will stop the convergecast, and send \mtag{Check} messages again when they detect a smaller initiator ID (line \ref{algl:send-check}).
This wrong convergecast can be executed at most $d$ times, where $d$ is the diameter of the initiator network at the time when all the initiators in the initiator network are in Phase 2.

\subsection{Rollback Algorithm}
\label{sec:rollback}

Here, we describe the rollback algorithm of CPS algorithm.
Actually, the algorithm is the same as \emph{RB algorithm} of SSS algorithm \cite{Moriya2001,Moriya2005}; thus, we just introduce RB algorithm in our style below.

First, we give the overview of RB algorithm.
The rollback algorithm can be invoked anytime by any node, even if some node in its snapshot was leaved from the system.
When a rollback of a snapshot is triggered by a rollback initiator $p_i$, first $p_i$ sends a \mtag{RbMarker} message to every node in $p_i$'s $DS$ to determine its rollback group similar to SSS algorithm described briefly in Section \ref{sec:sss-algorithm}.
After the rollback group is determined, each node in the group first restores its state to the latest checkpoint\footnote{If the node has not stored any checkpoint yet, the node rolls back to its initial state.} and recovers every link of the node with the stored in-transit messages.
Then, the node resumes to the execution of its application.

We enumerate the variables and the message types that RB algorithm uses below.
They are mostly the same for those of CPS algorithm.
In the rollback algorithm, each node $p_i$ maintains the following variables.
\begin{itemize}
	\setlength{\itemsep}{0em}
    \item $RbInit_i$: Rollback initiator's ID.
		An initiator sets this variable as its own ID.
		A normal node (not initiator) sets this variable to the initiator ID of the first \mtag{RbMarker} message it receives.
		Initially \KwNull.
    \item $RbRcvMk_i$: A set of the IDs of the nodes from which $p_i$ (already) received \mtag{RbMarker} messages.
		Initially $\emptyset$.
    \item $RbMkList_i$: A set of the IDs of the nodes from which $p_i$ has to receive \mtag{RbMarker} messages to terminate the algorithm.
		Initially $\emptyset$.
    \item $RbFin_i$: A boolean variable that denotes whether the rollback group is determined or not.
		Initially \KwFalse.
    \item $RbMkFrom_i$ (Initiator only): A set of the IDs of the nodes that send \mtag{RbMarker} to its DS.
		Initially $\emptyset$.
    \item $RbMkTo_i$ (Initiator only): The union set of the DSes of the nodes in $RbMkFrom$.
		Initially $\emptyset$.
    \item $RbDSInfo_i$ (Initiator only): A set of the pairs of a node ID and its DS.
		Initially $\emptyset$.
\end{itemize}
The algorithm also uses the following Phase 1 variables:
\begin{itemize}
	\setlength{\itemsep}{0em}
    \item $DS_i$
	\item $MsgQ_i$
\end{itemize}
We use the following message type for the rollback algorithm.
\begin{itemize}
	\setlength{\itemsep}{0em}
	\item \msg{\mtag{RbMarker}, init}: A message which controls the timing of a rollback of the local state.
		Parameter $init$ denotes the initiator's ID.

	\item \msg{\mtag{RbMyDS}, DS}: A message to send its own DS (all nodes communication-related to this node) to its initiator.
	\item \msg{\mtag{RbOut}}: A message to cancel the current rollback algorithm.
		When a node who has been an initiator receives a \mtag{RbMyDS} message of the node's previous instance, the node sends this message to cancel the sender's rollback algorithm instance.
	\item \msg{\mtag{RbFin}, List}: A message to inform that its rollback group is determined.
		$List$ consists of the IDs of the nodes from which the node has to receive \mtag{RbMarker} messages to terminate the algorithm.

\end{itemize}

Algorithm \ref{alg:rollback} is the pseudo-code of the rollback algorithm.
As you can see, this is mostly the same as Algorithm \ref{alg:phase1-normaloperations}, but the algorithm is simpler than that.
This is because the rollback algorithm does not support concurrent rollbacks of multiple groups, which requires collision handling of these groups like CPS algorithm.

\begin{algorithm}[t]
	\caption{Pseudo code of CPS algorithm for node $p_i$ (Rollback)}
\label{alg:rollback}
\begin{algorithmic}[1]

\algfontsize

\Procedure{Initiate}{ }
	\State OnReceive(\msg{\mtag{RbMarker}, p_i})
\EndProcedure

\Statex
\Procedure{OnReceive}{\msg{\mtag{RbMarker}, p_x} from $p_j$}
	\If{$RbInit_i = \KwNull$}
		\State Stop the execution of its application
		\State $RbInit_i \gets p_x$, $RbRcvMk_i \gets RbRcvMk_i \cup \{p_j\}$
		\State $RbMkList_i \gets \emptyset$, $RbFin_i \gets \KwFalse$
		\State Send \msg{\mtag{RbMyDS}, DS_i} to $RbInit_i$
		\State Send \msg{\mtag{RbMarker}, p_x} to $\forall p_k \in DS_i$
	\ElsIf{$RbInit_i = p_x$}
		\State $RbRcvMk_i \gets RbRcvMk_i \cup \{p_j\}$
		\If{$RbFin_i = \KwTrue$}
			\State $\proc{CheckRbTermination}()$
		\EndIf
	\EndIf
\EndProcedure

\Statex
\Procedure{OnReceive}{\msg{\mtag{RbMyDS}, DS_j} from $p_j$}
	\If{$RbInit_i = \KwNull \vee RbFin_i = \KwTrue$}
		\State Send \msg{\mtag{RbOut}} to $p_j$
	\Else
		\State $RbMkFrom_i \gets RbMkFrom_i \cup \{p_j\}$
		\State $RbMkTo_i \gets RbMkTo_i \cup \{DS_j\}$
		\State $RbDSInfo_i \gets RbDSInfo_i \cup (p_j, DS_j)$

		\If{$RbMkTo_i \subseteq RbMkFrom_i$}
			\State \Comment{Initiator $p_i$ determines its rollback group}
			\State $RbFin_i \gets \KwTrue$
			\For{each $p_k \in RbMkFrom_i$}
				\State $RbMkList_k \gets \{ \forall p_x \mid p_k \in DS_x, (p_x, DS_x) \in RbDSInfo_i\}$
				\State Send \msg{\mtag{RbFin}, RbMkList_k} to $p_k$
			\EndFor
		\EndIf

	\EndIf
\EndProcedure

\Statex
\Procedure{OnReceive}{\msg{\mtag{RbOut}} from $p_j$}
	\State \Comment{Cancel this rollback algorithm}
	\State $RbInit_i \gets \KwNull$
\EndProcedure

\Statex
\Procedure{OnReceive}{\msg{\mtag{RbFin}, List} from $p_j$}
	\State $RbMkList_i \gets List$
	\State \Comment{My initiator notifies the determination of its rollback group}
	\State $RbFin_i \gets \KwTrue$
	\State $\proc{CheckRbTermination}()$
\EndProcedure

\Statex
\Procedure{CheckRbTermination()}{}
	\If{$RbMkList_i \subseteq RbRcvMk_i$}
		\State Restore its state to the latest checkpoint
		\State Restore in-transit messages stored with the checkpoint to its links
		\For{each $(p_j, m)$ in $MsgQ_i$}

			\If{$p_j \notin RbMkList_i$}
				\State Add $m$ into the corresponding link
			\EndIf
		\EndFor
		\State Resume the execution of its application.
		\State Terminate this rollback algorithm
	\EndIf
\EndProcedure

\end{algorithmic}
\end{algorithm}

\color{black}

\section{Correctness}
\label{sec:correctness}

In this section, we show the correctness of the proposed algorithm.
First, we show the consistency of the recorded checkpoints (the snapshot).
The consistency of the snapshot can be guaranteed by the following conditions:
(a) the recorded checkpoints are mutually concurrent, which means that no causal relation, e.g., message communications, exists between any two checkpoints, and (b) in-transit messages are correctly recorded.

We denote the $k$-th event of node $p_i$ as $e^k_i$.
$S_i$ denotes the recorded checkpoint of node $p_i$.
When a snapshot algorithm correctly terminates, $S_i$ is updated to the latest checkpoint, and the previous recorded checkpoint is discarded.
Thus, $S_i$ is uniquely defined, if $p_i$ recorded its local state at least once.
From the proposed algorithm (and many other snapshot algorithms using \emph{Marker}), $S_i$ is usually created when the node receives the first \mtag{Marker}.

\begin{Define}
(A causal relation) $e^n_i \prec e^m_j$ denotes that $e^n_i$ causally precedes $e^m_j$.
This causal relation is generated in three cases:
(1) $e^n_i$ and $e^m_j$ are two internal computations on the same node ($i = j$) and $n < m$.
(2) $e^n_i$ and $e^m_j$ are the sending and the receiving events of a message, respectively.
(3) $e^n_i \prec e^l_k$ and $e^l_k \prec e^m_j$ (transitive).
\end{Define}

Now we show the following lemma using the notation and definition above.
\begin{Lemma}
\label{lem:checkpoint-concurrent}
For any two checkpoints $S_i$ and $S_j$ recorded at distinct nodes $p_i$ and $p_j$ by the proposed algorithm, neither $S_i \prec S_j$ nor $S_j \prec S_i$ holds (or they are concurrent).
\end{Lemma}
\begin{proof}
For contradiction, we assume $S_i \prec S_j$ holds without loss of generality.
It follows that a message chain $m_1, m_2, \cdots, m_k$ ($k \geq 1$) exists such that $m_1$ is sent by $p_i$ after $S_i$, $m_l$ is received before sending $m_{l+1}$ ($1 \leq l < k$) at a node, and $m_k$ is received by $p_j$ before $S_j$.

If $S_i$ and $S_j$ are recorded by \mtag{Markers} from the same initiator, we can show that \mtag{Marker} is sent along the same link before each $m_l$.
This is because \mtag{Marker} is (a) sent to every communication-related node when a node records a checkpoint, and (b) sent to a communication-irrelated node before a message is sent to the node (which becomes communication-related).
Therefore, $p_j$ records its checkpoint at the latest before it receives $m_k$, which is a contradiction.

Even if $S_i$ and $S_j$ are recorded by \mtag{Markers} from two different initiators, $p_x$ and $p_y$, respectively, \mtag{Marker} from $p_x$ is received by $p_j$ before 
the receipt of $m_k$ for the same reason as above.
Thus, $p_j$ never records its checkpoint, when \mtag{Marker} from $p_y$ is received by it (a collision occurs).

Therefore, Lemma \ref{lem:checkpoint-concurrent} holds.
\end{proof}

Next, we present the following lemma about the recorded in-transit messages.
\begin{Lemma}
\label{lem:in-transit-msg}
A message $m$ sent from $p_i$ to $p_j$ is recorded as an in-transit message by $p_j$, if and only if $m$ is sent before $S_i$ and received after $S_j$.
\end{Lemma}
\begin{proof}
\textbf{(only if part)}
A message $m$ from $p_i$ to $p_j$ is recorded as an in-transit message by $p_j$ only when it is received after $S_j$, but before \mtag{Marker} from $p_i$.
\mtag{Marker} is sent from $p_i$ to $p_j$ immediately after $S_i$; thus, the above implies from the FIFO property of the communication link that $m$ is sent before $S_i$.
The only if part holds.

\textbf{(if part)}
Let $m$ be the message that is sent from $p_i$ before $S_i$, and received by $p_j$ after $S_j$.
First, we assume that $S_i$ and $S_j$ are recorded on receipt of \mtag{Marker}s from the same initiator (i.e., they are in the same partial snapshot group).
Because $m$ is sent before $S_i$, $p_i$ adds $p_j$ to its $DS_i$, and then $p_i$ sends \mtag{Marker} to $p_j$ when $S_i$ is recorded (i.e., when the first \mtag{Marker} is received).
Node $p_i$ sends not only \mtag{Marker} but also its $DS_i$ to its initiator.
This implies when the snapshot group is determined, $p_i$ is included in $MkList_j$, which is the set of the IDs of the nodes from which $p_j$ has to receive \mtag{Markers}.
Therefore, $p_j$ cannot terminate the algorithm, until $p_j$ receives \mtag{Marker} from $p_i$.
Because $m$ is received by $p_j$ before \mtag{Marker} from $p_i$ (due to the FIFO property), $m$ is always recorded as an in-transit message.

Next, we assume that $S_i$ and $S_j$ are recorded on receipt of \mtag{Markers} from different initiators (denoted by $p_x$ and $p_y$, respectively).
In this case, when $p_j$ receives \mtag{Marker} from $p_i$ ($p_i$ has to send \mtag{Marker} to $p_j$ when it records $S_i$), it sends \mtag{NewInit} to its initiator $p_y$ because it detects a collision.
We have to consider the following two cases when $p_y$ receives \mtag{NewInit} from $p_j$.
Note that, at this time, $p_x$ has not determined its snapshot group, because $p_j$ is included in $DS_i$, and $p_x$ has not received $DS_j$ yet.

(1) \textbf{$p_y$ has not determined its snapshot group}:
$p_y$ sends \mtag{Link} to $p_x$, and a virtual link between the two nodes is created in the initiator network.
This causes $p_i$ to be added to $MkList_j$, when $p_y$ determines its snapshot group.
Because $p_i \in MkList_j$, $p_j$ has to wait for \mtag{Marker} from $p_i$, and records $m$ as an in-transit message.

(2) \textbf{$p_y$ already determined its snapshot group}:
If $p_i$ is in the snapshot group of $p_y$, we can show with an argument similar to (1) that $m$ is recorded as an in-transit message.
If $p_i$ is not in $p_y$'s snapshot group, then the snapshot group is determined using $DS_j$ that does not contain $p_i$.
This implies $p_j$ never sends \mtag{Marker} to $p_i$, when checkpoint $S_j$ is recorded.
In this case, because $p_y$ has already sent a \mtag{Fin} message to $p_j$ before the receipt of \mtag{NewInit}, $p_j$ never records $m$ in $S_j$, because $p_i$ is not included in $MkList_j$.
However, in this case, $p_j$ records a new checkpoint, say $S'_j$, on receipt of \mtag{Marker} from $p_i$ that was sent when $S_i$ is recorded, and receives $m$ before $S'_j$.
As a result, $m$ is not an in-transit message, and is never recorded in $S_j$ or $S'_j$.
\end{proof}

Lemmas \ref{lem:checkpoint-concurrent} and \ref{lem:in-transit-msg} guarantee the consistency of the recorded checkpoints and in-transit messages by the proposed algorithm.
Now we discuss about the termination of Phase 1 using the following lemma.

\begin{Lemma}
\label{lem:p1-termination}
Every initiator eventually terminates Phase 1 and proceeds to Phase 2.
\end{Lemma}

\begin{proof}
To terminate Phase 1 (and start Phase 2), each initiator has to execute procedure \proc{CanDetermineSG()} (lines \ref{algl:candeterminesg-begin} to \ref{algl:candeterminesg-end} in Algorithm \ref{alg:phase1-normaloperations}) and satisfies two conditions (line \ref{algl:candeterminesg-condition} in Algorithm \ref{alg:phase1-normaloperations}):
(1) $MkTo_i$ is a subset of or equal to $MkFrom_i$ and (2) $Wait_i$ is an empty set.
Note that whenever $MkTo_i$, $MkFrom_i$, or $Wait_i$ is updated, an initiator executes procedure \proc{CanDetermineSG()} (refer Algorithm \ref{alg:phase1-normaloperations}).
Therefore, if any initiator cannot terminate Phase 1, it implies that, two conditions are not satisfied and the variables in the two conditions are never updated (i.e., deadlock), or the two conditions are never satisfied forever even if they are repeatedly updated (i.e., livelock).

(1) \textbf{Condition $MkTo_i \subseteq MkFrom_i$}:
Assume for contradiction that $MkFrom_i \subset MkTo_i$ and no more update occurs.
Let $p_x$ be the node that is included in its initiator $p_i$'s $MkTo_i$, but not in $MkFrom_i$.
This means that $p_x$ received (or will receive) a \mtag{Marker} message from the node whose DS contains $p_x$.
When $p_x$ receives the \mtag{Marker} message, $p_x$ does one of the following (lines \ref{algl:marker-begin} to \ref{algl:marker-end} in Algorithm \ref{alg:phase1-normaloperations}):
(a) If it is the first \mtag{Marker} message (lines \ref{algl:p1-case-a-begin} to \ref{algl:p1-case-a-end}), $p_x$ sends its $DS_x$ to its initiator $p_i$, which is a contradiction.
(b) If it is the second or later \mtag{Marker} message (lines \ref{algl:p1-case-b-begin} to \ref{algl:p1-case-b-end}), $p_x$ already sent its $DS_x$ to its initiator $p_i$ when $p_x$ received the first \mtag{Marker} message, this is also a contradiction.
(c) If a collision happens (lines \ref{algl:p1-case-c-begin} to \ref{algl:p1-case-c-end}), we must take care with $MkFrom$ of two initiators, $p_x$'s initiator $p_i$ and the opponent collided initiator, say $p_j$.
For the initiator $p_i$, when $p_x$ receives a collided \mtag{Marker}, $p_x$ sends a \mtag{NewInit} message to its initiator $p_i$.
This implies that $p_x$ processed the case (a) to recognize $p_i$ as its initiator before, and the case (a) contradicts the assumption as we proved.
For the opponent initiator $p_j$, when $p_i$ receives the \mtag{NewInit} message, the initiator sends a \mtag{Link} message, which leads $p_x \in MkFrom_j$ (line \ref{algl:link-mkfrom-update} of Algorithm \ref{alg:collision-handling}).
This also contradics the assumption.

(2) \textbf{Condition $Wait_i = \emptyset$}:
Assume for contradiction that there is an element in $Wait_i$, and the element is never removed from $Wait_i$.
Note that an element can be added to $Wait_i$ only when a collision occurs for the first time between two snapshot groups (line \ref{algl:wait-add} in Algorithm \ref{alg:collision-handling}).
Therefore, when an initiator $p_i$ adds an element to $Wait_i$, $p_i$ also sends a \mtag{Link} message to the opponent colliding initiator $p_j$.
The initiator $p_j$ sends either an \mtag{Ack} message or a \mtag{Deny} message as its reply (lines \ref{algl:link-begin} to \ref{algl:link-end} in Algorithm \ref{alg:collision-handling}).
Both of these two messages cause the corresponding element to remove from $Wait_i$; thus, each element in $Wait_i$ is removed eventually.
This is a contradiction.
Note that if once two distinct initiators are connected in an initiator network by exchanging \mtag{Link} and \mtag{Ack} messages, they never add the opponent initiator in their $Wait$ each other.
If a \mtag{Deny} message is sent as the reply, the collision never occurs again between the two collided nodes.
Therefore, an element is added to $Wait_i$ only a finite number of times, because the total number of the nodes in the system is finite.
\end{proof}

From Lemmas \ref{lem:checkpoint-concurrent} to \ref{lem:p1-termination}, the following theorem holds.
\begin{Theorem}
\label{thm:p1-correctness}
	Phase 1 eventually terminates, and all checkpoints and in-transit messages recorded by the proposed algorithm construct a consistent snapshot of the subsystem.
\end{Theorem}

Now, we prove the following theorem regarding the correctness of Phase 2.
\begin{Theorem}
	\label{thm:p2-safety}
	Every initiator in an initiator network terminates, after all of the initiators in the network determine their snapshot groups.
\end{Theorem}

To prove the theorem, we will show that the convergecast in Phase 2 never terminates, if an initiaor executing Phase 1 exists.
The reason is as follows:
An initiator terminates Phase 2 when it receives a \mtag{GlobalTerm} message.
The root node of the spanning tree constructed on the initiator network sends \mtag{GlobalTerm} messages, when the node receives \mtag{LocalTerm} messages from all its neighbor nodes (they all are children of the node on the tree).
\mtag{LocalTerm} messages are sent by a convergecast from the leaf nodes of the tree to the root, when (1) a node received \mtag{Check} messages from all its neighbor nodes, and no neighbor node was a child of the node (or the node is a leaf), or (2) a node received \mtag{Check} messages from all its neighbor nodes and \mtag{LocalTerm} messages from all its child nodes.
Therefore, it is sufficient for the correctness of Phase 2 to prove the following lemma.

\begin{Lemma}
\label{lem:convergecast-safety}
The convergecast in Phase 2 never terminates, if an initiator node executing Phase 1 exists.
\end{Lemma}

\begin{proof}
We assume that only one node is executing Phase 1 in the initiator network, and let $p_i$ be the node.
We denote all nodes with distance $d$ from $p_i$ as $N^d_i$; e.g., $N^3_i$ is the set of all nodes with distance 3 from $p_i$ (trivially, $N^1_i = N_i$).
Let $p_s$ be the node that has the smallest ID in the initiator network.
To terminate the convergecast, $p_s$ must receive \mtag{LocalTerm} from all nodes in $N_s$ and become the root of the spanning tree.
Assuming that $p_s \in N_i$, the convergecast never terminates, because $p_i$ is executing Phase 1, and never sends \mtag{LocalTerm} to $p_s$.
Even if $p_s \in N^2_i$, the convergecast cannot terminate, because a node in $N_i$ that cannot receive \mtag{LocalTerm} from $p_i$ does not send \mtag{LocalTerm} to $p_s$.
In the same way, if $p_s \in N^x_i$ for some $x (\geq 1)$, the convergecast never terminates.
\end{proof}

If the convergecast does not terminate, which implies that an initiator is still executing Phase 1 and has not determined its snapshot group yet, no node can terminate Phase 2, because no \mtag{GlobalTerm} is sent.
Therefore, Theorem \ref{thm:p2-safety} holds.

\section{Evaluation}
\label{sec:evaluation}

In this section, we evaluate the performance of the proposed algorithm with CSS algorithm \cite{Kim2011,Kim2014}.
CSS algorithm is a representative of partial snapshot algorithms, as described in Section \ref{sec:related-works}, and the two algorithms have the same properties:
(1) The algorithms do not suspend an application execution on a distributed system while taking a snapshot, 
(2) the algorithms take partial snapshots (not snapshots of the entire system),
(3) the algorithms can take multiple snapshots concurrently, and
(4) the algorithms can handle dynamic network topology changes.
In addition, both algorithms are based on SSS algorithm \cite{Moriya2001,Moriya2005}.
For these reasons, CSS algorithm is a reasonable baseline for CPS algorithm.
We also analyze time and message complexities of CPS algorithm theoretically in Section \ref{sec:theoretical-performance}.

\subsection{CSS algorithm summary}
Before showing the simulation results, we briefly explain CSS algorithm.
For details, please refer the original paper \cite{Kim2014}.

The basic operation when no collision happens is almost the same as Phase 1 of CPS algorithm.
An initiator sends \mtag{Marker} messages to the nodes in its DS, and the nodes reply by sending \mtag{DSinfo} messages with their DS.
If the initiator receives DSes from all of its nodes, it sends \mtag{Fin} messages to let the nodes know the sets of nodes from which they must receive \mtag{Markers}, before terminating the snapshot algorithm.

In the algorithm, when a collision occurs, two collided initiators merge their snapshot groups into one group, and one of them becomes a main-initiator and the other becomes a sub-initiator.
The main-initiator manages all of the DSes of the nodes in the merged snapshot group and determines when the nodes terminate the snapshot algorithm.
The sub-initiator just forwards all the \mtag{DSinfo} and collision-related messages to its main-initiator, if it receives.
If another collision occurs and the main-initiator's snapshot group is merged into that of the merging initiator, the merged initiator resigns the main-initiator, and becomes a sub-initiator of the merging initiator.
These relations among a main-initiator and sub-initiators form a tree rooted at the main-initiator, and in this paper, we call it an \emph{initiator network}, like CPS algorithm.

Figure \ref{fig:css-collision-example} (a) illustrates the actual message flow of CSS algorithm when a collision happens.
When a node $p_x$ receives a collided \mtag{Marker} message from a neighbor node $p_y$, $p_x$ sends a \mtag{NewInit} message to its initiator.
This \mtag{NewInit} message is forwarded to the initiator's initiator if it exists.
This forwarding repeats until the \mtag{NewInit} message reaches the main-initiator.
The main-initiator $p_a$ sends an \mtag{Accept} message to $p_x$, to allow resolution of this collision.
Then, $p_x$ sends a \mtag{Combine} message to $p_y$, and this \mtag{Combine} message is also forwarded to the opponent main-initiator $p_b$.
When the opponent main-initiator $p_b$ receives the \mtag{Combine} message, the node compares its ID with ID of $p_a$.
If $p_a < p_b$, $p_b$ recognizes $p_a$ as its initiator, and sends an \mtag{InitInfo} message to $p_a$ with all of the information about the snapshot algorithm, including the set of all DSes that $p_b$ has ever received.
Otherwise, $p_b$ sends a \mtag{CompInit} message to $p_a$ and requests $p_a$ to become $p_b$'s sub-initiator, by considering $p_b$ as its main-initiator.
The collision is resolved with these message exchanges, and finally, one of the initiators $p_a$ or $p_b$ manages both snapshot groups.
When $p_b$ becomes the main-initiator by sending the \mtag{CompInit} message, the initiator network of this example can be illustrated as in Figure \ref{fig:css-collision-example} (b).

When another collision happens during this collision handling, the main initiator stores the \mtag{NewInit} message that provides the notification of the collision in a temporary message queue, and processes the message after the current collision is resolved.
In other words, CSS algorithm can handle at most one collision at the same time.
We think this drawback largely degrades the performance of CSS algorithm.

\begin{figure}[tbp]
	\centering
	\begin{tabular}{cc}
		\includegraphics[scale=0.55]{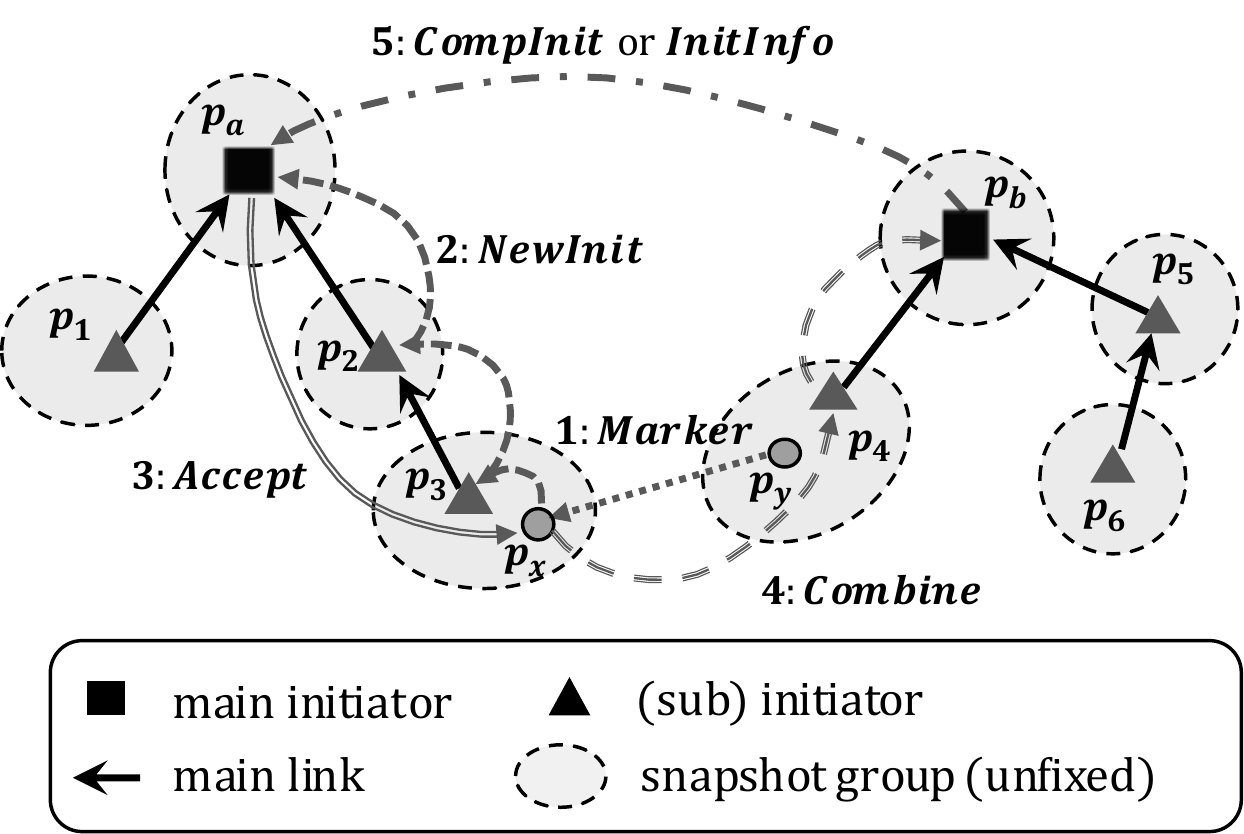} &
		\includegraphics[scale=0.35]{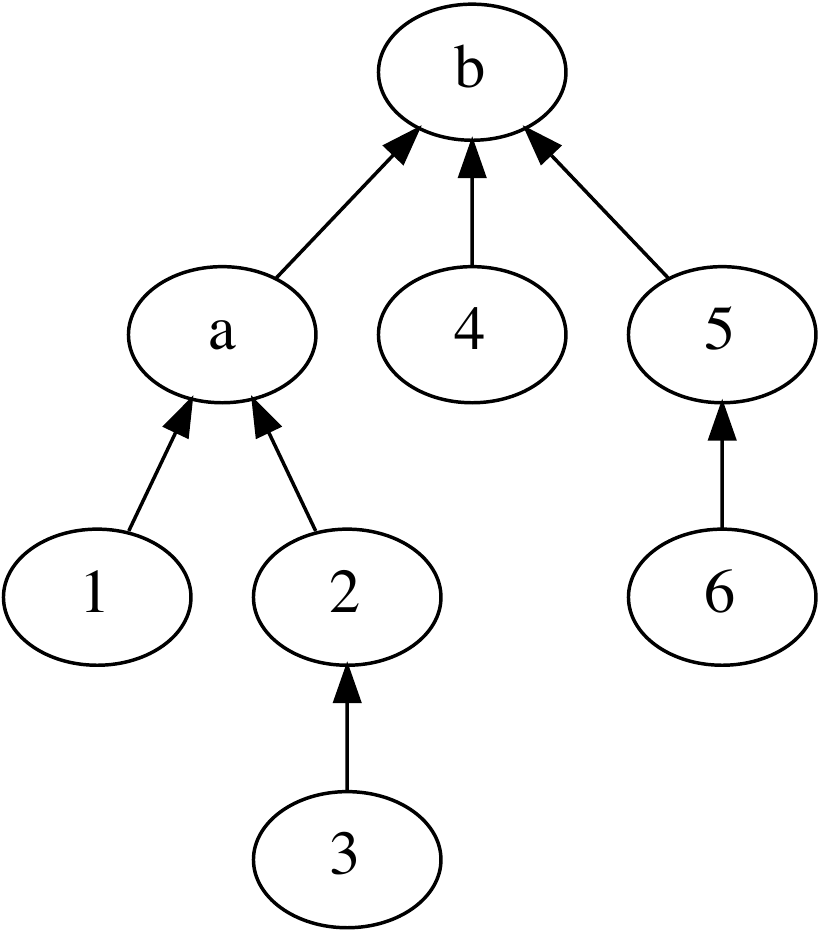} \\
		(a) Message flow when a collision occurs&
		(b) Initiator network\\
	\end{tabular}
	\caption{A collision-handling example of CSS algorithm}
	\label{fig:css-collision-example}
\end{figure}

In the simulation, we modified CSS algorithm slightly from the original, because we discovered during implementing the simulator that the original algorithm lacked some mechanisms that were necessary to take snapshots consistently.
First, we introduced \mtag{Out} messages, which was not described in CSS algorithm paper \cite{Kim2014}.
This helps a node (not an initiator) to shut down the current snapshot algorithm and join the next one.
Second, we altered it to forward \mtag{CompInit} and \mtag{InitInfo} messages to a main-initiator, in addition to \mtag{DSinfo} and \mtag{Combine}.
This was necessary to avoid deadlocking, when two or more collisions occur at the same time.

\subsection{Simulation settings}
\label{sec:simulation-settings}

The evaluation is performed by simulating node behaviors on a single computer.
Although both algorithms can take a snapshot on an asynchronous distributed system, for simplicity, a simulation is conducted in synchronous rounds.
In a round, all nodes receive messages, process them, and send new messages, which will be delivered in the next round.

Before each simulation of the algorithms, a communication-relation on nodes is generated, which has influence on the performance of the snapshot algorithms.
Although actual communication-relations depend on distributed applications to which snapshot algorithms are applied, we generate communication-relations randomly with probability C for every pair of nodes for simplicity.
After generating a communication-relation, we start simulation executions, one of each of the algorithms.
In the first round, each node becomes an initiator with probability F, and starts execution (by storing its state and sending \mtag{Markers} to its communication-related nodes) of the snapshot algorithms if it becomes an initiator.
We terminate the simulation when all the initiated snapshot algorithm instances terminate.

We have three parameters for the simulation: communication probability $C$, initiation probability $F$, and the number of nodes $N$.
As described, parameters $C$ and $F$ probabilistically determine the communication-relations and the snapshot algorithm initiations, respectively.
The larger $C$ generates denser communication-relations; thus, a (partial) snapshot group becomes larger.
The larger $F$ makes more nodes behave as initiators.
$N$ indicates the number of nodes in a simulation.
If $C$ or $F$ is large, a collision occurs more easily.

We evaluate these snapshot algorithms with three measures.
The first measure is the total number of messages sent in a simulation.
As described in Section \ref{sec:system-model}, a node can send a message to any other node if the node knows the destination node's ID.
Additionally, in this simulation, we assume that every node can send messages (including messages sent in Phase 2 of CPS algorithm, e.g., \mtag{Check}) to every other node in one hop.
In other words, we do not take into account any relaying message for this measure.
The second measure is the total number of rounds from the initiations of the snapshot algorithms until the termination of all snapshot algorithm instances.
The last measure is the number of messages by type.
This is a complement of the first measure, to discuss which parts of the algorithms dominate their communication complexity.
For this purpose, we classify the messages of both algorithms into four types, as shown in Table \ref{tab:msg-type}.
The normal-type messages are used to decide a snapshot group.
The collision-type messages are sent to resolve collisions that occurred during a snapshot algorithm.
The initiator network-type messages are sent between initiators, to coordinate their instances.
In CPS algorithm, this type of message is used in Phase 2, to synchronize their termination.
In contrast, CSS algorithm uses this type to forward collision-related messages from a sub-initiator to its main-initiator.

We run at least 100 simulations for each parameter setting and show the average of the simulations.

\begin{table}[tbp]
	\centering
	\caption{Message types.
		The initiator network-type messages of CSS algorithm (i.e., \mtag{DSinfo}, \mtag{NewInit}, etc.) are counted only when these messages are forwarded from a sub-initiator to its main-initiator.}
	\label{tab:msg-type}
	\small
	\begin{tabular}{lll}
\hline
\textbf{Type} & \textbf{CPS algorithm} & \textbf{CSS algorithm} \\ \hline\hline
Marker & \mtag{Marker} & \mtag{Marker} \\
Normal & \mtag{MyDS}, \mtag{Fin}, \mtag{Out} & \mtag{DSinfo}, \mtag{Fin}, \mtag{Out} \\
Collision & \mtag{NewInit}, \mtag{Link}, \mtag{Ack}, \mtag{Deny}, \mtag{Accept} & \mtag{NewInit}, \mtag{Accept}, \mtag{Combine}, \mtag{CompInit}, \mtag{InitInfo}\\
Initiator network & \mtag{Check}, \mtag{LocalTerm}, \mtag{GlobalTerm} & \mtag{DSinfo}, \mtag{NewInit}, \mtag{Combine}, \mtag{CompInit}, \mtag{InitInfo} \\
\hline
	\end{tabular}
\end{table}

\subsection{Simulation results}

First, we show the simulation results for different numbers of nodes $N$, in Figure \ref{fig:node-results}.
As Figure \ref{fig:node-results} (a) indicates, CPS algorithm can take snapshots with fewer messages than CSS algorithm.
For instance, when $N=200$, CPS algorithm reduced 44.1\% of messages from that of CSS algorithm.
Figure \ref{fig:node-results} (b) shows the running time of these algorithms (note that only this graph uses a logarithmic scale).
Although the running time of CPS algorithm was always less than 40 rounds, that of CSS algorithm drastically increased, and it took 34,966 rounds when $N=200$.
This huge difference came from the fact that CSS algorithm can handle at most one collision at the same time; thus, collisions must wait until the collision being processed (if it exists) is resolved.
In contrast, an initiator of CPS algorithm can handle multiple collisions concurrently, and then CPS algorithm drastically improves the total rounds.
We discuss later why the huge differences in the total numbers of messages and rounds exist.

The total number of collisions of both algorithms are displayed in Figure \ref{fig:node-results} (c).
Interestingly, CPS algorithm has more collisions than CSS algorithm, although CPS algorithm sends fewer messages than CSS algorithm.
This is because, CPS algorithm reprocesses a \mtag{Marker} message again when a node receives \mtag{Out} to resolve a collision consistently.
However, if the node is in another snapshot group than that of the \mtag{Marker} message, this reprocess leads to a collision.

Figure \ref{fig:node-results} (d) shows the total numbers of partial snapshot groups\footnote{These are equal to the numbers of initiators}, which are controlled by initiation probability $C$.
Both the algorithms have the same numbers because we provided the same seed of the pseudo random number generator (PRNG) in the simulator to each iteration of both the algorithms; we used $i$ as the seed for the $i$-th iteration of each algorithm.
Moreover, the initiation of each node is calculated with the PRNG in the same manner between the algorithms; thus, the same set of nodes become initiators for the same iteration.

Figure \ref{fig:node-results} (e) depicts the size of their initiator networks in the simulations.
Here, we define the initiator network size of CPS algorithm and CSS algorithm by the diameter of the initiator network and the depth of the initiator network tree, respectively, because these metrics can estimate the message processing load of the initiator network.
We can observe that the increasing ratio of CSS algorithm is larger than that of CPS algorithm.

Figures \ref{fig:node-results} (f) and (g) display the ratio of the message types, which were defined in Section \ref{sec:simulation-settings}, of the algorithms in their simulations.
The ratios of marker-type messages of the two algorithms are mostly the same, while those of collision- and initiator network-type messages are different.
In CPS algorithm, Initiator network-type messages are sent on the initiator network only to construct a breath-first-search (BFS) spanning tree, and to synchronize the termination of the initiators' instances.
However, CSS algorithm requires sub-initiators to forward every collision-related message, in which these forwarding messages are counted as initiator network-type messages, to their main-initiators.
This forwarding is a very heavy task in terms of the message counts.
In fact, 40.9\% of messages were sent on the initiator network of CSS algorithm when $N=200$, although the total numbers of collision-type messages are mostly the same for the algorithms.

\begin{figure}[tbp]
	\centering
	\begin{tabular}{c}
		\includegraphics[scale=0.77]{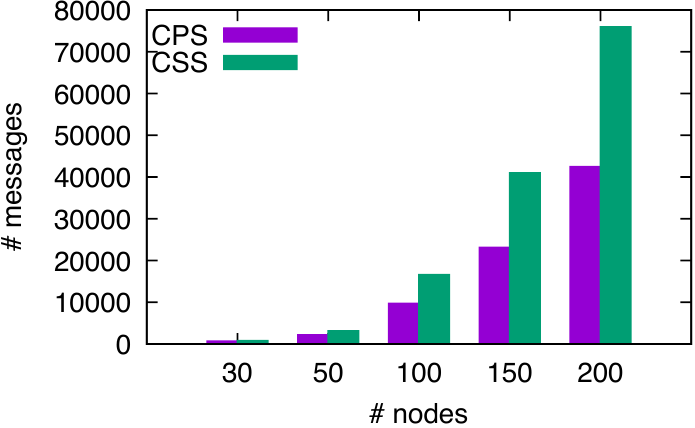} \\
		(a) Total messages \\
		\includegraphics[scale=0.77]{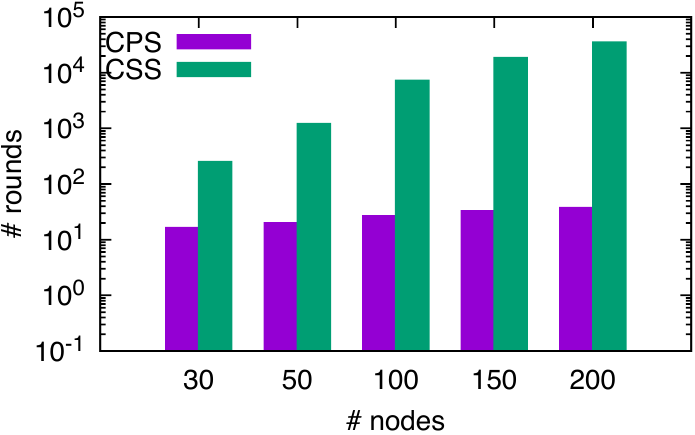} \\
		(b) Total rounds \\
		\includegraphics[scale=0.77]{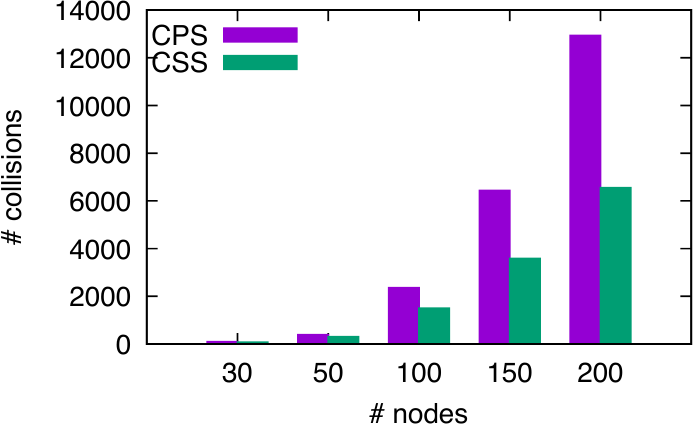} \\
		(c) Total collisions \\
	\end{tabular}
	\begin{tabular}{cc}
		\includegraphics[scale=0.77]{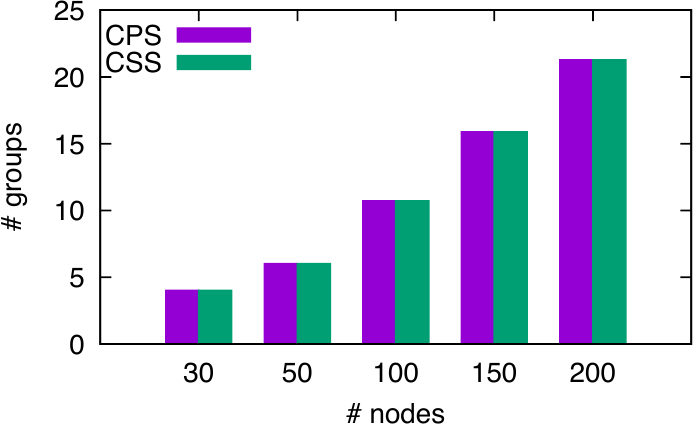} &
		\includegraphics[scale=0.77]{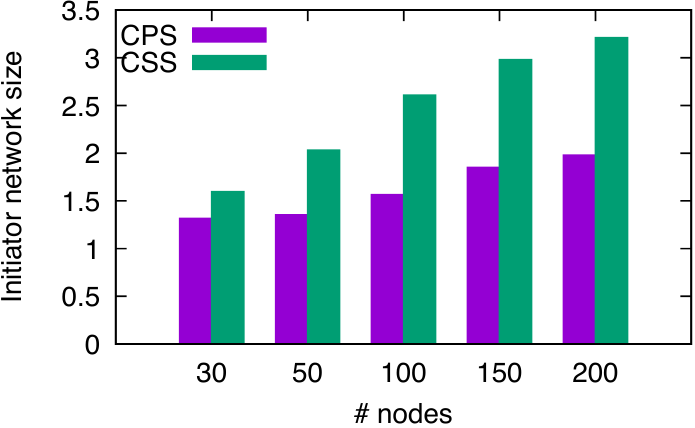} \\
		(d) Total partial snapshot groups & (e) Initiator network size \\
		\includegraphics[scale=0.77]{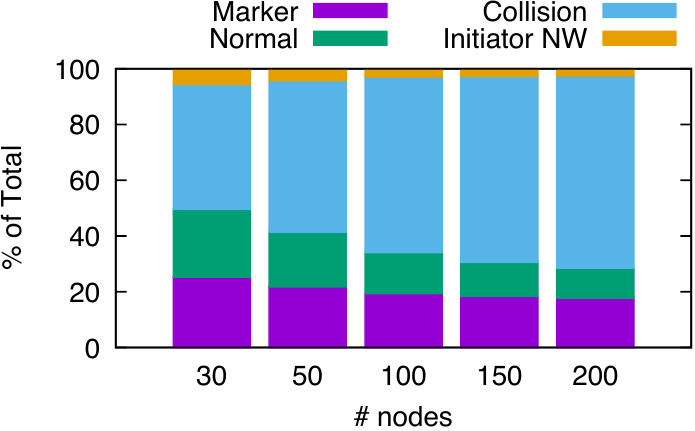} &
		\includegraphics[scale=0.77]{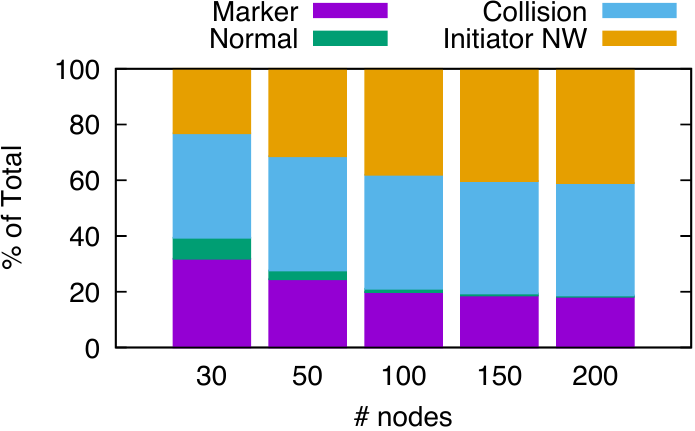} \\
		(f) Ratio of CPS messages & (g) Ratio of CSS messages \\
	\end{tabular}
	\caption{Simulation results for different numbers of nodes $N$.
		Communication probability $C$ and initiation probability $F$ are fixed at 10\%}
	\label{fig:node-results}
\end{figure}

To discuss why there exist such huge differences in the total numbers of messages and rounds between CPS algorithm and CSS algorithm, we examine their representative executions, and analyze their execution details.
As the representative, we chose an execution whose total number of messages is almost the same as the average of each algorithm when $N=200$, $C=10$, and $F=10$.

First, we see the BFS spanning tree on the initiator network of CPS algorithm in the execution, which is illustrated in Figure \ref{fig:cps-init-nw}.
There are 17 initiators in the network, and its topology is almost a complete graph (the network has a clique of size 16, and its diameter is two).
Therefore, the convergecast in Phase 2 with \mtag{Check} messages terminates at most two rounds after all the initiators finish Phase 1, and the root node can broadcast \mtag{GlobalTerm} immediately.
We can confirm this in Figure \ref{fig:node-results} (d), and this is not a special case for the execution.

\begin{figure}[tbp]
	\centering
	\includegraphics[scale=0.35]{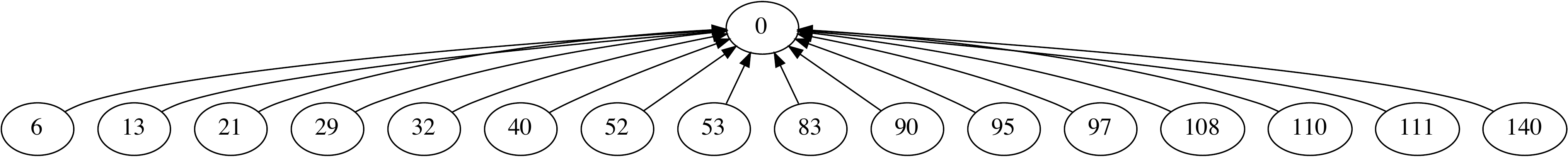}
	\caption{An initiator network example of CPS algorithm}
	\label{fig:cps-init-nw}
\end{figure}

The initiator network of CSS algorithm is depicted in Figure \ref{fig:css-init-nw}.
The tree has 16 nodes (initiators), and its depth is five, which means a collision-related message (e.g., \mtag{Combine} or \mtag{NewInit}) will forward four times at most.
To reveal the reason for the large number of messages and rounds of CSS algorithm, let us assume that a \mtag{Marker} message is sent from the snapshot group of initiator $p_{173}$ to the snapshot group of initiator $p_{171}$, and this tree has been constructed when this collision happens.
This is the worst case on the network.
First, the collided node in $p_{171}$'s snapshot group sends a \mtag{NewInit} message to $p_{171}$, and this message is forwarded four times to $p_0$; then $p_0$ sends an \mtag{Accept} message to $p_{171}$.
When $p_{171}$ receives this \mtag{Accept} message, it sends a \mtag{Combine} message to the colliding node in $p_{173}$'s snapshot group, and this \mtag{Combine} message is also forwarded four times to $p_{0}$.
\footnote{
	Remember that the initiator network in Fig.~\ref{fig:css-init-nw} has been constructed when this collision happens.
	This means that $p_0$ is the main-initiator of both $p_{173}$ and $p_{171}$.
	In other words, $p_0$ behaves as the main-initiator of the collided snapshot group and as that of the colliding snapshot group.
}
Then, $p_0$ receives the \mtag{Combine} message from $p_0$, and $p_0$ replies with an \mtag{InitInfo} message to $p_0$, because $p_0 \not< p_0$.
Finally, the collision between the initiators that share the same parent is resolved, thanks to 12 messages and 12 rounds (remember, the simulation is conducted by synchronous round, and it always takes a round to deliver a message).
Moreover, CSS algorithm must resolve collisions one by one.
Although this is a worst-case analysis, and typically, CSS algorithm can handle a collision with fewer messages and rounds, this is why CSS algorithm consumes a large number of messages and rounds.

\begin{figure}[tbp]
	\centering
	\includegraphics[scale=0.35]{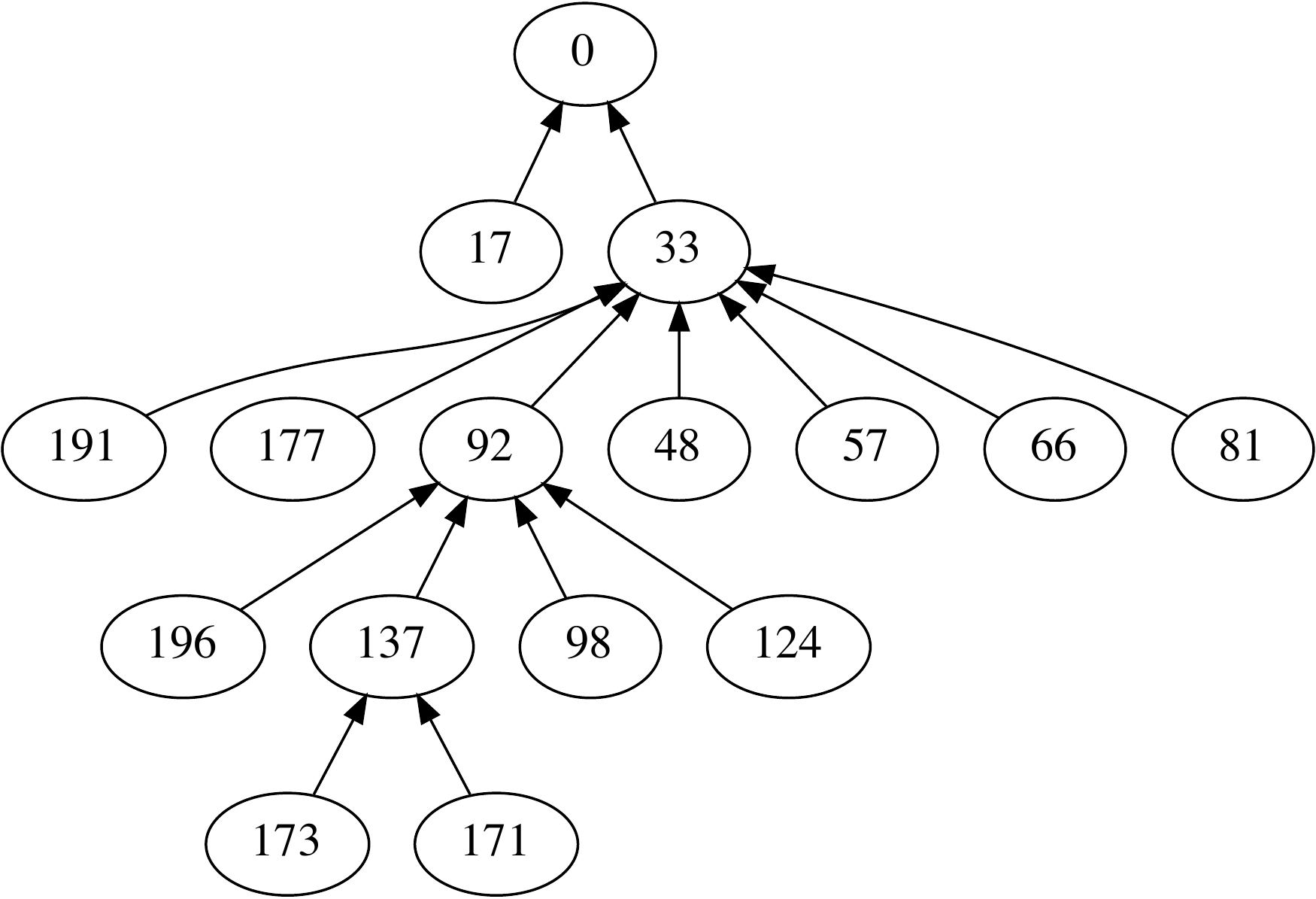}
	\caption{An initiator network example of CSS algorithm}
	\label{fig:css-init-nw}
\end{figure}

Figure \ref{fig:mc_histogram} shows the top 10 nodes that process the largest number of messages in the two executions of CSS algorithm and CPS algorithm.
Apparently, most of the messages in CSS algorithm are processed by two nodes ($p_0$ and $p_{33}$ in Figure \ref{fig:css-init-nw}).
This is unfavorable, because the nodes are exhausted by processing these messages, and can no longer run an application.
However, these tasks are distributed equally in CPS algorithm.

\begin{figure}[tbp]
	\centering
	\includegraphics[scale=0.77]{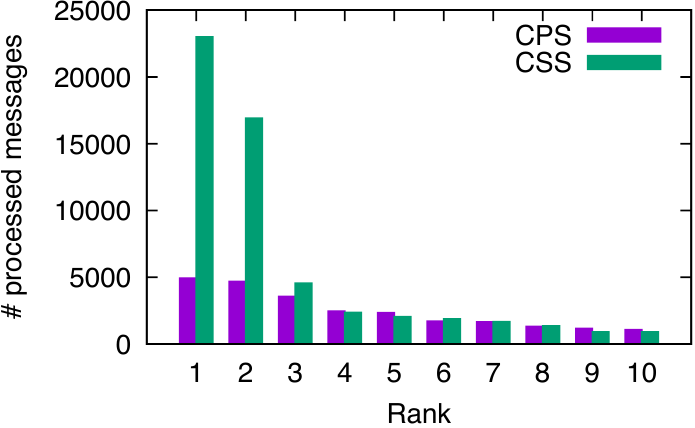}
	\caption{The total number of processed messages of the top 10 nodes in the simulation}
	\label{fig:mc_histogram}
\end{figure}

Finally, we observe the results for different communication probability $C$ and initiation probability $F$.
These results are shown in Figures \ref{fig:com-freq-results} and \ref{fig:init-freq-results}.
Similarly to the case for different number of $N$, CPS algorithm outperforms CSS algorithm in terms of the total numbers of messages and rounds.

\begin{figure}[tbp]
	\centering
	\begin{tabular}{cc}
		\includegraphics[scale=0.77]{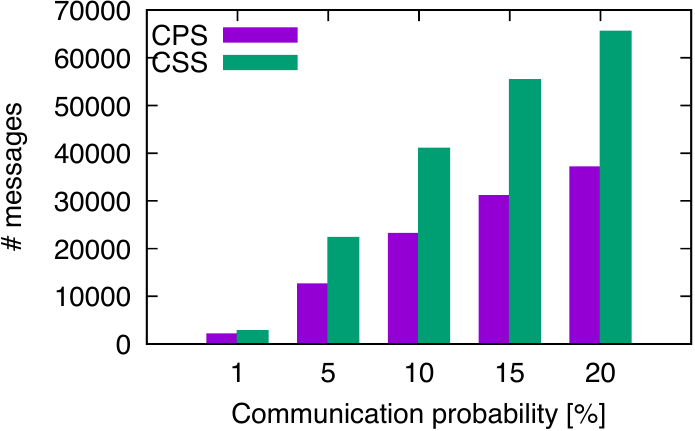} &
		\includegraphics[scale=0.77]{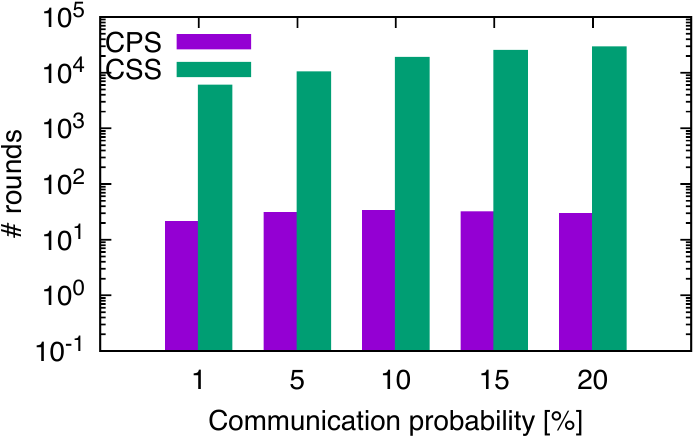} \\
		(a) Total messages & (b) Total rounds \\
	\end{tabular}
	\caption{Simulation results for different communication probability $C$.
		The number of nodes $N$ and initiation probability $F$ are fixed at 150 and 10\%, respectively}
	\label{fig:com-freq-results}
\end{figure}

\begin{figure}[tbp]
	\centering
	\begin{tabular}{cc}
		\includegraphics[scale=0.77]{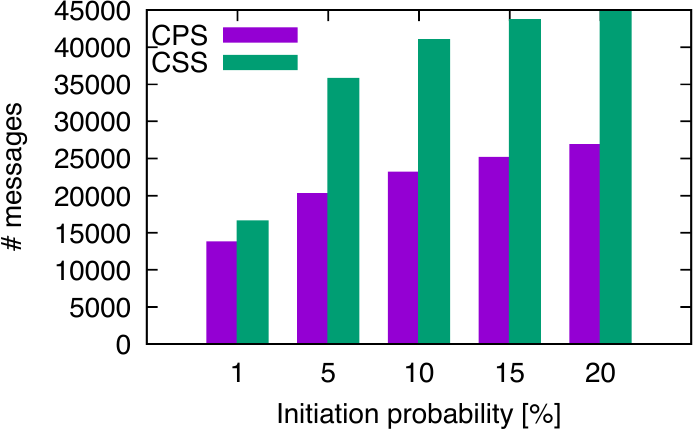} &
		\includegraphics[scale=0.77]{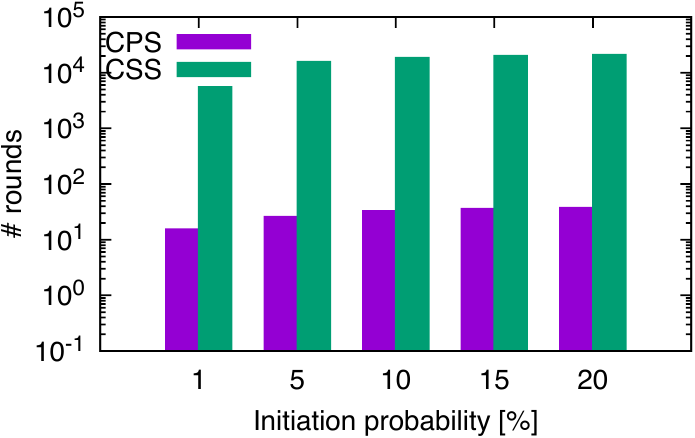} \\
		(a) Total messages & (b) Total rounds \\
	\end{tabular}
	\caption{Simulation results for different initiation probability $F$.
		The number of nodes $N$ and communication probability $C$ are fixed at 150 and 10\%, respectively}
	\label{fig:init-freq-results}
\end{figure}

\subsection{Theoretical Performance}
\label{sec:theoretical-performance}

Finally, we analyze the theoretical performance of CPS algorithm in terms of time and message complexities in the worst scenario where there are $n$ nodes in the system, and all of them invoke the algorithm.
We also assume the invocations happen at the same time for simplicity.

First, we analyse the time complexity with \emph{asynchronous rounds}.
In an asynchronous round, every node receives messages sent in the previous round, processes the messages, and sends new messages to other nodes.
We assume that communication-relations of all the nodes form a line graph of $n$ nodes, and one end of the graph has the smallest ID for the worst case of time complexity.
In this case, each initiator determines its partial snapshot group in five rounds\footnote{Each initiator sends messages in the following order: \mtag{Marker} (round 1), \mtag{MyDS} and \mtag{NewInit} (round 2), \mtag{Link} (round 3), \mtag{Ack} (round 4), and \mtag{Accept} (round 5).}, and enters Phase 2.
The leader election of Phase 2 takes $n-1$ rounds because it requires $n-1$ rounds to propagate the smallest ID from one end to the other end on the line graph.
With the same discussion, the relay transmissions of \mtag{LocalTerm} and \mtag{GlobalTerm} messages also takes $n-1$ rounds each.
After the termination of Phase 2, each initiator sends \mtag{Fin} messages and terminates CPS algorithm in the next round.
Therefore, CPS algorithm can take a snapshot within $3n+3$ rounds.

Next, we consider message complexity of CPS algorithm.
The worst case is a situation where all the initiators are communication-related each other.
In Phase 1 of the case, each node sends $n$ \mtag{Marker} messages and one \mtag{MyDS} message before collisions happen.
Since a collision requires four messages and $n$ collisions happen in this situation, $4n$ messages are sent to resolve the collisions in total.
In the leader election process of Phase 2, $m$ \mtag{Check} messages are sent in a round, and the election finish within $\Delta$ rounds, where $m$ is the number of edges in the initiator network, and $\Delta$ is the diameter of the network when Phase 2 terminates.
\mtag{LocalTerm} and \mtag{GlabalTerm} messages are sent once in every edge; then the total number of these messages is $m$.
Since we assume in Phase 1 that collisions happen between every two initiators, the initiator network is a complete graph of degree $n$, that is, $m=n(n-1)/2$ and $\Delta=1$.
Therefore, the message complexity of CPS algorithm is $\mathcal{O}(n^2)$.

\section{Conclusion}
\label{sec:conclusion}
We proposed a new partial snapshot algorithm named CPS algorithm to realize efficient checkpoint-rollback recovery in large-scale and dynamic distributed systems.
The proposed partial snapshot algorithm can be initiated concurrently by two or more initiators, and an overlay network among the initiators is constructed to guarantee the consistency of the snapshot obtained when some snapshot groups overlap.
CPS algorithm realizes termination detection to consistently terminate the algorithm instances that are initiated concurrently.

In a simulation, we confirmed that the proposed CPS algorithm outperforms the existing partial snapshot algorithm CSS in terms of the message and time complexities.
The simulation results showed that the message complexity of CPS algorithm is better than that of CSS algorithm for all the tested situations, e.g., 44.1\% better when the number of nodes in a distributed system is 200.
This improvement was mostly due to the effective use of the initiator network.
The time complexity was also drastically improved, because CPS algorithm can handle multiple collisions concurrently, while CSS algorithm must handle collisions sequentially.

\subsection*{Acknowledgements}
This work was supported by JSPS KAKENHI Grant Numbers JP16K16035, JP18K18029, and JP19H04085.
All the experiments in the paper were conducted with GNU Parallel \cite{Tange2018} on the supercomputer of ACCMS, Kyoto University.


\begin{thebibliography}{10}

\bibitem{Kim2018}
Y.~Kim, J.~Nakamura, Y.~Katayama, and T.~Masuzawa, ``{A Cooperative Partial
  Snapshot Algorithm for Checkpoint-Rollback Recovery of Large-Scale and
  Dynamic Distributed Systems},'' in {\em Proceedings of the 6th International
  Symposium on Computing and Networking Workshops (CANDARW)}, (Takayama,
  Japan), pp.~285--291, Nov. 2018.

\bibitem{Nakamura2020}
J.~Nakamura, Y.~Kim, Y.~Katayama, and T.~Masuzawa, ``{A cooperative partial
  snapshot algorithm for checkpoint-rollback recovery of large-scale and
  dynamic distributed systems and experimental evaluations},'' {\em Concurrency
  and Computation: Practice and Experience}, vol.~n/a, p.~e5647, Jan. 2020.

\bibitem{Koo1987}
R.~Koo and S.~Toueg, ``Checkpointing and rollback-recovery for distributed
  systems,'' {\em IEEE Transactions on Software Engineering}, vol.~SE-13,
  pp.~23--31, Jan 1987.

\bibitem{Netzer1995}
R.~H. Netzer and J.~Xu, ``Necessary and sufficient conditions for consistent
  global snapshots,'' {\em IEEE Transactions on Parallel \& Distributed
  Systems}, vol.~6, pp.~165--169, 02 1995.

\bibitem{Fischer1982}
M.~J. Fischer, N.~D. Griffeth, and N.~A. Lynch, ``Global states of a
  distributed system,'' {\em IEEE Transactions on Software Engineering},
  vol.~SE-8, pp.~198--202, May 1982.

\bibitem{Briatico1984}
D.~Briatico, A.~Ciuffoletti, and L.~Simoncini, ``A distributed domino-effect
  free recovery algorithm.,'' in {\em Proceedings of the 4th Symposium on
  Reliability in Distributed Software and Database Systems}, pp.~207--215, 1984.

\bibitem{Spezialetti1986}
M.~Spezialetti and P.~Kearns, ``Efficient distributed snapshots,'' in {\em
  Proceedings of the 6th International Conference on Distributed Computing
  Systems (ICDCS)}, pp.~382--388, 1986.

\bibitem{Prakash1994}
R.~Prakash and M.~Singhal, ``Maximal global snapshot with concurrent
  initiators,'' in {\em Proceedings of the 6th IEEE Symposium on Parallel and
  Distributed Processing}, pp.~344--351, 1994.

\bibitem{Elnozahy2002}
E.~N. Elnozahy, L.~Alvisi, Y.-M. Wang, and D.~B. Johnson, ``A survey of
  rollback-recovery protocols in message-passing systems,'' {\em ACM Computing
  Surveys}, vol.~34, pp.~375--408, Sept. 2002.

\bibitem{Moriya2005}
S.~Moriya and T.~Araragi, ``Dynamic snapshot algorithm and partial rollback
  algorithm for internet agents,'' {\em Electronics and Communications in Japan
  (Part III: Fundamental Electronic Science)}, vol.~88, no.~12, pp.~43--57,
  2005.

\bibitem{Kim2011}
Y.~Kim, T.~Araragi, J.~Nakamura, and T.~Masuzawa, ``Brief announcement: a
  concurrent partial snapshot algorithm for large-scale and dynamic distributed
  systems,'' in {\em Proceedings of the 13th international conference on
  Stabilization, Safety, and Security of distributed systems (SSS)}, SSS'11,
  (Grenoble, France), pp.~445--446, Oct. 2011.

\bibitem{Chandy1985}
K.~M. Chandy and L.~Lamport, ``Distributed snapshots: Determining global states
  of distributed systems,'' {\em ACM Trans. Comput. Syst.}, vol.~3, pp.~63--75,
  Feb. 1985.

\bibitem{Lai1987}
T.~H. Lai and T.~H. Yang, ``On distributed snapshots,'' {\em Information
  Processing Letters}, vol.~25, no.~3, pp.~153--158, 1987.

\bibitem{Kshemkalyani2010}
A.~D. Kshemkalyani, ``Fast and message-efficient global snapshot algorithms for
  large-scale distributed systems,'' {\em IEEE Transactions on Parallel and
  Distributed Systems}, vol.~21, pp.~1281--1289, Sep. 2010.

\bibitem{Garg2006}
R.~Garg, V.~K. Garg, and Y.~Sabharwal, ``Scalable algorithms for global
  snapshots in distributed systems,'' in {\em Proceedings of the 20th Annual
  International Conference on Supercomputing}, ICS '06, pp.~269--277, 2006.

\bibitem{Garg2010}
R.~Garg, V.~K. Garg, and Y.~Sabharwal, ``Efficient algorithms for global
  snapshots in large distributed systems,'' {\em IEEE Transactions on Parallel
  and Distributed Systems}, vol.~21, pp.~620--630, May 2010.

\bibitem{Helary1999}
J.~Helary, A.~Mostefaoui, and M.~Raynal, ``Communication-induced determination
  of consistent snapshots,'' {\em IEEE Transactions on Parallel and Distributed
  Systems}, vol.~10, no.~9, pp.~865--877, 1999.

\bibitem{Baldoni1998}
R.~Baldoni, F.~Quaglia, and B.~Ciciani, ``A vp-accordant checkpointing protocol
  preventing useless checkpoints,'' in {\em Proceedings of the 17th IEEE
  Symposium on Reliable Distributed Systems}, pp.~61--67, 1998.

\bibitem{Baldoni1997}
R.~Baldoni, J.-M. Helary, A.~Mostefaoui, and M.~Raynal, ``A
  communication-induced checkpointing protocol that ensures rollback-dependency
  trackability,'' in {\em Proceedings of IEEE 27th International Symposium on
  Fault Tolerant Computing}, pp.~68--77, 1997.

\bibitem{Moriya2001}
S.~Moriya and T.~Araragi, ``Dynamic snapshot algorithm and partial rollback
  algoithm for internet agents,'' in {\em Proceedings of the 15th International
  Symposium on Distributed Compuiting (DISC 2001)}, pp.~23--28, 2001.

\bibitem{Kim2014}
Y.~Kim, T.~Araragi, J.~Nakamura, and T.~Masuzawa, ``{A Concurrent Partial
  Snapshot Algorithm for Large-scale and Dynamic Distributed Systems},'' {\em
  IEICE Transactions on Information and Systems}, vol.~E97-D, pp.~65--76, Jan.
  2014.

\bibitem{Tange2018}
O.~Tange, {\em GNU Parallel 2018}.
\newblock Zenodo, first~ed., 2018.

\end{thebibliography}
\end{document}